\documentclass[aps,prx,twocolumn,groupedaddress,amsmath,superscriptaddress]{revtex4-2}

\usepackage{amsfonts}
\usepackage[dvips]{graphicx}
\usepackage{color}
\usepackage{amssymb}
\usepackage{hyperref}
\usepackage{color}
\usepackage{mathrsfs}
\usepackage{cleveref}
\usepackage{lipsum}
\usepackage{isomath}
\usepackage{amsmath,amsthm}
\usepackage{outlines}
\usepackage{epstopdf}
\usepackage{txfonts}
\DeclareMathAlphabet{\mathcal}{OMS}{cmsy}{m}{n}
\usepackage{dsfont}
\usepackage{setspace}
\usepackage{multirow}
\usepackage{booktabs}
\hypersetup{
	colorlinks=true,     
	urlcolor=blue,       
	linkcolor=blue,      
	citecolor=blue
}
\usepackage[T1]{fontenc}
\usepackage{colortbl}
\usepackage{hhline}
\usepackage{xcolor}
\usepackage[normalem]{ulem}
\allowdisplaybreaks[4]
\usepackage{enumitem}

\makeatletter
\definecolor{nblue}{rgb}{0.3,0.3,1.0}%229
\definecolor{ngreen}{rgb}{0.2,0.7,0.2}%161
\definecolor{nred}{rgb}{0.9,0.1,0}%711&900
\definecolor{nblack}{rgb}{0,0,0}
\definecolor{nyellow}{rgb}{1.0,0.75,0.0}
\hypersetup{
colorlinks=true,  
linkcolor=blue,  
citecolor=blue,  
urlcolor=blue,    
pdfborder={0 0 0}, 
}

\usepackage{array}
\usepackage{tabularx}
\usepackage[english]{babel}
\newcommand{\bra}[1]{\langle{#1}|}
\newcommand{\ket}[1]{|{#1}\rangle}

\newcommand{\beq}{\begin{equation}}
\newcommand{\eeq}{\end{equation}}
\newcommand{\bqa}{\begin{eqnarray}}
\newcommand{\eqa}{\end{eqnarray}}

\newtheorem{lemma}{Lemma}

\crefname{lemma}{Lemma}{Lemmas}

\makeatletter

\usepackage{microtype}
\makeatletter
\def\fnum@figure{FIG.~\thefigure}
\makeatother

\begin{document}

\title{Semidefinite-programming hierarchies for classically simulable state sets}
\author{Mengyan Li}
\affiliation{School of Mathematical Sciences, Beijing University of Posts and Telecommunications, Beijing 100876, China}
\affiliation{Key Laboratory of Mathematics and Information Networks, Beijing University of Posts and Telecommunications, Ministry of Education, Beijing 100876, China}
\affiliation{State Key Laboratory of Networking and Switching Technology, Beijing University of Posts and Telecommunications, Beijing 100876, China}

\author{Yanning Jia}
\affiliation{School of Mathematical Sciences, Beijing University of Posts and Telecommunications, Beijing 100876, China}
\affiliation{Key Laboratory of Mathematics and Information Networks, Beijing University of Posts and Telecommunications, Ministry of Education, Beijing 100876, China}
\affiliation{State Key Laboratory of Networking and Switching Technology, Beijing University of Posts and Telecommunications, Beijing 100876, China}

\author{Fenzhuo Guo}\email{gfenzhuo@bupt.edu.cn}
\affiliation{School of Mathematical Sciences, Beijing University of Posts and Telecommunications, Beijing 100876, China}
\affiliation{Key Laboratory of Mathematics and Information Networks, Beijing University of Posts and Telecommunications, Ministry of Education, Beijing 100876, China}
\affiliation{State Key Laboratory of Networking and Switching Technology, Beijing University of Posts and Telecommunications, Beijing 100876, China}

\author{Haifeng Dong}
\affiliation{School of Instrumentation Science and Opto-Electronics Engineering, Beihang University, Beijing 100191, China}

\author{Sujuan Qin}
\affiliation{State Key Laboratory of Networking and Switching Technology, Beijing University of Posts and Telecommunications, Beijing 100876, China}

\author{Fei Gao}
\affiliation{State Key Laboratory of Networking and Switching Technology, Beijing University of Posts and Telecommunications, Beijing 100876, China}

\begin{abstract}

Determining whether a quantum state set admits a classical explanation is a fundamental problem in quantum information theory, as it reveals whether the information encoded in state preparation can be reproduced by classical means. Here we study classically simulable state sets, namely those admitting realizations using shared randomness and preparation devices that individually generate only pairwise commuting state sets. We develop a complete semidefinite programming (SDP) hierarchy to characterize the set of classically simulable state sets. The key step is a reformulation of classical simulability as a feasibility problem involving deterministic response functions and auxiliary positive-operator-valued measures (POVMs) simulable by rank-one projective measurements. We then establish a complete SDP hierarchy for rank-one projectively simulable POVMs, which in turn induces computable outer approximations converging to the set of classically simulable state sets. The corresponding dual SDP further yields affine witnesses that can be used to rule out classical simulability. Applying the hierarchy to state sets under depolarizing noise yields computable upper bounds on the critical visibility. For several highly symmetric examples, these bounds are matched by explicit classical simulations, providing evidence for the practical usefulness of the hierarchy.
\end{abstract}

\maketitle

\section{Introduction}

Quantum state preparation is a fundamental primitive in quantum information processing. The ability to prepare a prescribed set of quantum states underlies a wide range of applications, such as quantum key distribution~\cite{2014bennettQuantumCryptography,2009scaraniSecurityPractical}, quantum state discrimination~\cite{1969helstromQuantumDetection,2009barnettQuantumState}, quantum random access codes~\cite{1999ambainisDenseQuantum,2002ambainisDenseQuantuma,2015tavakoliQuantumRandom}, etc. The operational capabilities of these tasks depend largely on the nonclassical relations among the quantum states in the set. A state set with these relations constitutes a genuine quantum resource only if it admits no classical explanation. Indeed, whenever such an explanation exists, the information encoded in the preparation process can be represented by classical random variables, and no genuine quantum advantage can arise in the corresponding task~\cite{2021designolleSetCoherence,2026cobucciOperationallyClassical}. This naturally leads to the problem of determining whether a given quantum state set admits a classical explanation or instead constitutes a genuine quantum resource for state preparation.

Designolle \emph{et al.} introduced the notion of set coherence to characterize the nonclassical relations within quantum state sets~\cite{2021designolleSetCoherence}. In finite dimensions, a state set has vanishing set coherence if its members commute pairwise. Such sets are termed classical state sets, since all states are diagonal in a common orthonormal basis and the preparation information is encoded solely in a probability distribution over the basis elements. Noncommutativity, however, does not rule out a classical explanation of state preparation. Recently, Cobucci \emph{et al.} proposed the notion of classical simulability~\cite{2026cobucciOperationallyClassical}, which adopts a more general notion of classicality. A state set is classically simulable if it can be realized using shared randomness together with preparation devices that individually generate only classical state sets. This implies that a state set may have nonvanishing set coherence while remaining classically simulable. Moreover, no classically simulable state set can certify irreducible nonclassicality in prepare-and-measure scenarios, since all observable statistics it generates admit a classical hidden-variable model~\cite{2026cobucciOperationallyClassical}. These observations identify classical simulability as the appropriate criterion for determining whether a state set admits a classical explanation.

The set of classically simulable state sets is convex, as it is precisely the convex hull of classical state sets. Ref.~\cite{2026cobucciOperationallyClassical} performed numerical searches over finite classical state sets to construct explicit classical simulations, providing sufficient certificates of classical simulability. Nevertheless, a systematic characterization of the set of classically simulable state sets remains unavailable. Similar characterization problems arise throughout quantum information theory, where semidefinite programming (SDP)~\cite{boydConvexOptimization2004,2024tavakoliSemidefiniteProgramming} has become one of the principal tools for analyzing convex quantum structures. Notable examples include the Doherty--Parrilo--Spedalieri hierarchy for separable states~\cite{2002dohertyDistinguishingSeparable,2004dohertyCompleteFamily}, the Navascu\'es--Pironio--Ac\'in hierarchy for Bell correlations~\cite{2007navascuesBoundingSet,2008navascuesConvergentHierarchy}, and a recently proposed hierarchy for projective measurement simulability~\cite{2026brinsterRobustCertification}. These developments demonstrate the effectiveness of SDP hierarchies in characterizing intricate convex sets.

In this work, we develop a complete SDP hierarchy for characterizing the set of classically simulable state sets. We first reformulate classical simulability as a feasibility problem over deterministic response functions and auxiliary positive-operator-valued measures (POVMs) simulable by rank-one projective measurements, thereby establishing a direct connection between classical simulability of state sets and projective simulability of POVMs. Building on this connection, we construct a complete SDP hierarchy for rank-one projectively simulable POVMs, which in turn induces a complete hierarchy for the set of classically simulable state sets. The corresponding affine witnesses, derived from the dual programs, can be used to rule out classical simulability. As applications, we investigate the classical simulability of state sets under depolarizing noise and derive upper bounds on the corresponding critical visibilities. Our numerical results on several symmetric state sets show that the resulting upper bounds are consistent with explicit classical simulations, illustrating the potential utility of the hierarchy for assessing classical simulability.

\section{Classically Simulable State Sets}

Let $\boldsymbol{\sigma}=\{\sigma_x\}_{x=1}^{n}$ be a set of quantum states on a $d$-dimensional Hilbert space. We call $\boldsymbol{\sigma}$ classical if its members commute pairwise, i.e., $[\sigma_x,\sigma_{x'}]=0$ for all $x,x'\in[n]:=\{1,2,\ldots,n\}$. We say that $\boldsymbol{\rho}=\{\rho_x\}_{x=1}^{n}$ is classically simulable if there exists a probability distribution $p(\lambda)$ and, for each $\lambda$, a classical state set $\boldsymbol{\sigma}^{\lambda}=\{\sigma_x^{\lambda}\}_{x=1}^{n}$ such that
\begin{equation}
	\label{equ:classical1}
	\rho_x
	=
	\int d\lambda\,
	p(\lambda)\,
	\sigma_x^\lambda,
	\quad
	\forall x\in[n].
\end{equation}

We denote by $\mathcal{C}(d,n)$ the set of all classically simulable $n$-state sets in dimension $d$. The set $\mathcal C(d,n)$ is precisely the convex hull of classical state sets and hence forms a convex set. For each $\lambda$, since $\boldsymbol{\sigma}^{\lambda}$ is classical, there exists a common orthonormal basis $\{|e_i^\lambda\rangle\}_{i=1}^{d}$ such that 
$
\sigma_x^\lambda
=
\sum_{i=1}^{d}
\bar q(i|x,\lambda)
|e_i^\lambda\rangle\langle e_i^\lambda|
$
for all $x\in[n]$. Conditioned on $\lambda$, the state set is therefore fully characterized by the conditional probability distribution $\{\bar{q}(i|x,\lambda)\}_{i,x}$, with no coherence in the basis $\{|e_i^\lambda\rangle\}_{i}$. Since the set of conditional probability distributions forms a convex polytope whose extremal points are the deterministic response functions, one can further write $\bar q(i|x,\lambda)=\sum_{\mu=1}^{d^n}q(\mu|\lambda)D(i|x,\mu)$, where $\mu$ labels a deterministic assignment and $D(i|x,\mu) = \delta_{i,\mu(x)}$. Substituting these decompositions into Eq.~\eqref{equ:classical1}, one obtains
\begin{equation}
	\label{equ:classical2}
	\rho_x
	=
	\sum_{\mu=1}^{d^n}
	\sum_{i=1}^{d}
	D(i|x,\mu)\,
	\tau_{i,\mu},
	\quad
	\forall x,
\end{equation}
where $\tau_{i,\mu}:=\int d\lambda\,p(\lambda)q(\mu|\lambda)\,
|e_i^\lambda\rangle\langle e_i^\lambda|$. By construction, $\tau_{i,\mu} \succeq 0$ and $
\mathrm{Tr}(\tau_{i,\mu})
=
\int d\lambda\,p(\lambda)q(\mu|\lambda)
=:c_\mu
$ is independent of $i$, with $c_\mu\ge0$ and $\sum_\mu c_\mu=1$. For every $\mu$ satisfying $c_\mu>0$, writing $\tau_{i,\mu}=c_\mu\tau_{i|\mu}$, the normalized operators $\{\tau_{i|\mu}\}_{i=1}^{d}$ define a $d$-outcome POVM. As discussed in Appendix~\ref{appA}, each such POVM is simulable by rank-one projective measurements~\cite{2017oszmaniecSimulatingPositiveoperatorvalued}. Denoting by $\mathcal{P}(d,d;1)$ the set of all $d$-outcome POVMs on a $d$-dimensional Hilbert space that are simulable by rank-one projective measurements, we obtain $\{\tau_{i|\mu}\}_{i=1}^{d}\in\mathcal{P}(d,d;1)$ for every $\mu$ with $c_\mu>0$.

This establishes a direct connection between classical simulability of state sets and projective simulability of POVMs. Accordingly, the task of deciding whether a given set $\boldsymbol{\rho}$ belongs to $\mathcal{C}(d,n)$ can be expressed as the following feasibility problem:
\begin{equation}
	\label{equ:feasibilityproblem}
	\begin{aligned}
		\max\quad & 0\\
		\mathrm{w.r.t.}\quad& \{\tau_{i,\mu}\}_{i,\mu},\, \{c_\mu\}_\mu \\
		\mathrm{s.t.}\quad & \rho_x=\sum_{\mu=1}^{d^n}\sum_{i=1}^d D(i|x,\mu)\,\tau_{i,\mu},\quad \forall\, x,\\
		& \{\tau_{i|\mu}\}_{i=1}^d \in \mathcal{P}(d,d;1),\quad \forall\, \mu\in\mathcal N_{+},\\
		& \tau_{i,\mu}\succeq 0,\quad \mathrm{Tr}(\tau_{i,\mu})=c_\mu,\quad \forall\, i,\mu,\\
		& c_\mu\ge0,\quad \sum_{i=1}^d \tau_{i,\mu} = c_\mu \mathbb{I}_d, \quad \forall\, \mu,\quad \sum_{\mu=1}^{d^n} c_\mu = 1.
	\end{aligned}
\end{equation}
Here, $\mathcal N_{+}:=\{\mu:c_\mu>0\}$. If a semidefinite characterization of $\mathcal P(d,d;1)$ is available, then the above feasibility problem can be cast as an SDP by passing to the associated conic hull~\cite{boydConvexOptimization2004}
\[
\operatorname{cone} \bigl(\mathcal{P}(d,d;1)\bigr) := \left\{ \{T_i\}_{i=1}^d \;\middle|\; \begin{aligned}
	&\exists\, c\ge 0,\ \{M_i\}_{i=1}^d\in\mathcal{P}(d,d;1),\\
	& T_i=c\,M_i,\quad i=1,\dots,d
\end{aligned} \right\}.
\]
Introducing the conic hull eliminates the nonlinear normalization  $\tau_{i|\mu}=\tau_{i,\mu}/c_\mu$ appearing in Eq.~\eqref{equ:feasibilityproblem}. Indeed, for each $\mu\in\mathcal N_+$, the condition $\{\tau_{i|\mu}\}_{i=1}^{d}\in\mathcal P(d,d;1)$ is equivalent to requiring the corresponding unnormalized operators $\{\tau_{i,\mu}\}_{i=1}^{d}$ belong to $\operatorname{cone}(\mathcal P(d,d;1))$. In this formulation, the normalization parameter $c_\mu$ is absorbed into the conic variables, so that the membership constraint becomes homogeneous and no division by $c_\mu$ is required. 

To investigate the effect of depolarizing noise on classical simulability, we introduce the depolarizing channel $\mathrm\Lambda_v(\rho)=v\rho+[(1-v)/d]\mathbb{I}_d$, with $v\in[0,1]$, and define the critical visibility
\begin{equation}
	\label{equ:classical_visibility}
	\chi(\boldsymbol{\rho})
	=
	\max\left\{
	v\in[0,1]
	\;\middle|\;
	\{\mathrm\Lambda_v(\rho_x)\}_x\in\mathcal C(d,n)
	\right\}.
\end{equation}
The quantity $\chi(\boldsymbol{\rho})$ is the largest visibility such that the depolarized state set belongs to $\mathcal C(d,n)$. In particular,
$\chi(\boldsymbol{\rho})=1$
if and only if
$\boldsymbol{\rho}\in\mathcal C(d,n)$.

\section{Semidefinite characterization of $\mathcal{P}(d,d;1)$ and completeness of the hierarchy}

To characterize the set of rank-one projectively simulable POVMs, we introduce a hierarchy of semidefinite outer approximations to $\mathcal{P}(d,d;1)$, building on the hierarchy construction of Ref.~\cite{2026brinsterRobustCertification}. For each integer $\ell\ge2$, let
$\mathcal{P}_\ell(d,d;1)$ denote the set of POVMs
$\mathbf M=\{M_i\}_{i=1}^{d}$
with $\mathrm{Tr}(M_i)=1$ for all $i$,
for which there exist lifted operators
$\{\mathrm{\Pi}_{\mathbf i}^{(\ell)}\}_{\mathbf i\in[d]^\ell}$
acting on $(\mathbb C^d)^{\otimes\ell}$, where
$\mathbf i=(i_1,\ldots,i_\ell)$, satisfying
\begin{equation}
	\label{equ:liftedconstraint}
	\left\{
	\begin{aligned}
		&\mathrm{\Pi}_{\mathbf i}^{(\ell)}\succeq0,
		\quad\mathrm{Tr}\!\left(
		\mathrm{\Pi}_{\mathbf i}^{(\ell)}
		\right)=1,
		\quad\forall\,\mathbf i\in[d]^\ell,\\
		&\sum_{\mathbf i\setminus i_j}
		\mathrm{\Pi}_{\mathbf i}^{(\ell)}
		=
		\mathbb I_{d^{j-1}}
		\otimes
		M_{i_j}
		\otimes
		\mathbb I_{d^{\ell-j}},
		\quad
		\forall\,j\in[\ell],\,i_j\in[d],\\
		&
		\mathrm{Tr}_{\alpha}
		\!\left(
		V_{(\alpha\beta)}
		\mathrm{\Pi}_{\mathbf i}^{(\ell)}
		\right)
		=
		0,
		\quad
		\forall\,
		\mathbf i,\,
		\alpha<\beta:
		i_\alpha\neq i_\beta.
	\end{aligned}
	\right.
\end{equation}
Here
$\sum_{\mathbf{i}\setminus i_j}$
denotes summation over all entries of $\mathbf{i}$ except the $j$th one,
$V_{(\alpha\beta)}
=
\sum_{r,s=1}^{d}
|r\rangle_\beta\langle s|
\otimes
|s\rangle_\alpha\langle r|
\otimes
\mathbb I_{\mathrm{rest}}$
swaps the $\alpha$th and $\beta$th tensor factors, and
$\mathrm{Tr}_{\alpha}$ denotes the partial trace over the $\alpha$th tensor factor. We start the hierarchy at $\ell=2$ because, at $\ell=1$, the lifted operator is simply
$\mathrm{\Pi}_i^{(1)}=M_i$, and no swap constraints are present. Thus, the level-1 formulation is trivial and does not provide a nontrivial relaxation of projective simulability. The sets
$\{\mathcal P_\ell(d,d;1)\}_{\ell\ge2}$
define a sequence of semidefinite outer approximations to
$\mathcal P(d,d;1)$.
We first show that these relaxations form a hierarchy.

\textit{Proposition 1 (Hierarchy structure).}
For every $\ell\ge2$,
\[
\mathcal P_{\ell+1}(d,d;1)
\subseteq
\mathcal P_\ell(d,d;1).
\]

\textit{Proof.}
Let $\mathbf M\in\mathcal P_{\ell+1}(d,d;1)$ with feasible lifted operators $\{\mathrm{\Pi}_{\mathbf i,i_{\ell+1}}^{(\ell+1)}\}$ at level $\ell+1$, where $\mathbf i=(i_1,\ldots,i_\ell)$. Following Ref.~\cite{2026brinsterRobustCertification}, define
$\mathrm{\Pi}_{\mathbf i}^{(\ell)}
:=
d^{-1}
\sum_{i_{\ell+1}=1}^{d}
\mathrm{Tr}_{\ell+1}
\!\bigl(
\mathrm{\Pi}_{\mathbf i,i_{\ell+1}}^{(\ell+1)}
\bigr)$.
Positivity and trace-one are preserved under partial trace and convex averaging. For the marginal constraints, for every $j\in[\ell]$,
$
\sum_{\mathbf i\setminus i_j}
\mathrm{\Pi}_{\mathbf i}^{(\ell)}
=
d^{-1}
\mathrm{Tr}_{\ell+1}
\!\left(
\sum_{(\mathbf i,i_{\ell+1})\setminus i_j}
\mathrm{\Pi}_{\mathbf i,i_{\ell+1}}^{(\ell+1)}
\right)
=
\mathbb I_{d^{j-1}}
\otimes
M_{i_j}
\otimes
\mathbb I_{d^{\ell-j}},
$
where the last equality follows from the level-$(\ell+1)$ feasibility conditions. The swap constraints follow similarly, since
$V_{(\alpha\beta)}$
acts only on the first $\ell$ tensor factors and therefore commutes with
$\mathrm{Tr}_{\ell+1}$.
Hence $\{\mathrm{\Pi}_{\mathbf i}^{(\ell)}\}_{\mathbf i}$ satisfies all constraints in Eq.~\eqref{equ:liftedconstraint}, establishing that $\mathbf M\in\mathcal P_\ell(d,d;1)$.
\hfill$\square$

The hierarchy is moreover asymptotically complete.

\hypertarget{theorem1}{}
\textit{Theorem 1 (Completeness of the rank-one projective simulability hierarchy).} Let
$\mathbf M=\{M_i\}_{i=1}^{d}$
be a POVM on a $d$-dimensional Hilbert space satisfying
$\mathrm{Tr}(M_i)=1$
for all $i$.
Then
$\mathbf M\in\mathcal P_\ell(d,d;1)$
for every
$\ell\ge2$
if and only if
$\mathbf M\in\mathcal P(d,d;1)$.

\textit{Sketch of Proof.}
The ``if'' direction is immediate: any convex mixture of rank-one projective measurements yields feasible lifted operators at every hierarchy level via the corresponding tensor-product extension. For the ``only if'' direction, assume that $\mathbf M\in\mathcal P_\ell(d,d;1)$ for every $\ell\ge2$. The feasible lifted operators can then be assembled into an exchangeable set of classical--quantum (CQ) states. By the quantum de Finetti theorem~\cite{2002cavesUnknownQuantum}, every fixed $k$-copy marginal of the resulting infinite exchangeable set admits a representation as a convex mixture of independent and identically distributed $k$-copy CQ states. The one-copy marginal identifies $\mathbf M$ as a convex mixture of POVMs in the support of the de Finetti measure. The swap constraints in Eq.~\eqref{equ:liftedconstraint}, applied to the two-copy marginals, force distinct POVM effects to be pairwise orthogonal almost everywhere with respect to the de Finetti measure. Finally, the constraint $\mathrm{Tr}(\mathrm{\Pi}_{\mathbf{i}}^{(\ell)})=1$  implies that the effects in the support have unit trace and hence rank one. Therefore, $\mathbf M\in\mathcal P(d,d;1)$. Detailed proofs are given in Appendix~\ref{appB}.
\hfill$\square$

\section{Semidefinite Characterization of $\mathcal{C}(d,n)$ and the Dual Witness}

We now return to the central problem of the paper: deciding whether a given set $\boldsymbol{\rho}$ belongs to $\mathcal{C}(d,n)$. Eq.~\eqref{equ:classical2} shows that this question can be reduced to the existence of a simulation over deterministic response functions and positive semidefinite operators $\{\tau_{i,\mu}\}$ whose normalized versions $\{\tau_{i|\mu}\}_{i=1}^d$ form a POVM in $\mathcal{P}(d,d;1)$ for every $\mu\in\mathcal N_+$. In this way, deciding classical simulability of state sets is reduced to deciding membership of a collection of auxiliary POVMs in $\mathcal{P}(d,d;1)$. 

Theorem~\hyperlink{theorem1}{1} therefore makes this reduction effective: the complete hierarchy for $\mathcal{P}(d,d;1)$ induces a complete hierarchy for $\mathcal{C}(d,n)$. Accordingly, for each hierarchy level $\ell$, we define $\mathcal{C}_\ell(d,n)$ by replacing the exact constraint $\{\tau_{i|\mu}\}_{i=1}^d\in\mathcal{P}(d,d;1)$ in Eq.~\eqref{equ:feasibilityproblem} with its level-$\ell$ relaxation. For brevity, we shall refer collectively to all constraints in Eq.~\eqref{equ:liftedconstraint} as the lifted constraints. Passing to the associated conic form then leads to the following SDP feasibility problem:
\begin{equation}
	\label{equ:mainsdp} 
	\begin{aligned} 
		\max\quad & 0\\ \mathrm{w.r.t.}\quad& {\{\tau_{i,\mu}\}_{i,\mu},\, \{c_\mu\}_\mu,\, \{\mathrm{\Pi}^{(\ell)}_{\mathbf{i},\mu}\}_{\mathbf{i},\mu}} \\ \mathrm{s.t.}\quad & \rho_x =\sum_{\mu=1}^{d^n}\sum_{i=1}^dD(i|x,\mu)\,\tau_{i,\mu},\quad \forall\, x,\\ & \mathrm{\Pi}^{(\ell)}_{\mathbf{i},\mu}\succeq 0,\quad \mathrm{Tr}(\mathrm{\Pi}^{(\ell)}_{\mathbf{i},\mu})=c_\mu,\quad \forall\,\mathbf{i},\mu,\\ & \sum_{\mathbf{i}\setminus i_j}\mathrm{\Pi}^{(\ell)}_{\mathbf{i},\mu} =\mathbb{I}_{d^{j-1}}\otimes \tau_{i_j,\mu}\otimes \mathbb{I}_{d^{\ell-j}},\quad \forall\, j, i_j,\mu,\\ & \mathrm{Tr}_\alpha\!\left(V_{(\alpha\beta)}\mathrm{\Pi}^{(\ell)}_{\mathbf{i},\mu}\right)=0,\quad \forall\,\mu,\mathbf{i},\alpha<\beta:i_\alpha\neq i_\beta,\\ & \tau_{i,\mu}\succeq 0,\quad \mathrm{Tr}(\tau_{i,\mu})=c_\mu,\quad \forall\,i,\mu,\\ & c_\mu\ge 0,\quad \sum_{i=1}^d \tau_{i,\mu} = c_\mu \mathbb{I}_d, \quad\forall\,\mu,\quad \sum_{\mu=1}^{d^n} c_\mu=1. 
	\end{aligned}
\end{equation}
\noindent Here the lifted constraints are imposed in conic form for every $\mu$. This is equivalent to imposing the constraints only for $\mu\in\mathcal N_+$, since $c_\mu=0$ forces $\tau_{i,\mu}=0$ and hence $\mathrm\Pi^{(\ell)}_{\mathbf{i},\mu}=0$, rendering the corresponding lifted constraints trivial. The hierarchy defined by Eq.~\eqref{equ:mainsdp} is referred to as the classical simulability hierarchy. A computationally cheaper alternative, referred to as the aggregated classical simulability hierarchy (or simply the aggregated hierarchy), is obtained by first forming the aggregated operators $\bar{\tau}_i:=\sum_\mu\tau_{i,\mu}$ and imposing the lifted constraints only on the single POVM $\{\bar{\tau}_i\}_{i=1}^d$. This construction removes the deterministic assignment index from the lifted variables, substantially reducing the SDP size at the cost of discarding part of the hidden-variable structure. Consequently, feasibility of the classical simulability hierarchy implies feasibility of the aggregated hierarchy, but not conversely. Accordingly, the classical simulability hierarchy yields tighter level-$\ell$ outer approximations, whereas the aggregated hierarchy provides a computationally cheaper preliminary test. The explicit SDP formulations and complexity estimates for both hierarchies are given in Appendix~\ref{appC}.

\hypertarget{theorem2}{}
\textit{Theorem 2 (Completeness of the classical simulability hierarchy).} 
Let $\boldsymbol{\rho}=\{\rho_x\}_{x=1}^{n}$ be a state set on a $d$-dimensional Hilbert space. Then $\boldsymbol{\rho}\in\mathcal C_\ell(d,n)$ for every $\ell\ge2$ if and only if $\boldsymbol{\rho}\in\mathcal C(d,n)$.

The completeness of the classical simulability hierarchy follows directly from Theorem~\hyperlink{theorem1}{1}. Since Eq.~\eqref{equ:mainsdp} introduces no relaxation beyond replacing $\mathcal P(d,d;1)$ by $\mathcal P_\ell(d,d;1)$, exact recovery of $\mathcal P(d,d;1)$ in the hierarchy limit transfers directly to exact recovery of $\mathcal C(d,n)$. Theorem~\hyperlink{theorem2}{2} therefore establishes a convergent semidefinite characterization of classical simulability: the set $\mathcal C(d,n)$ is recovered exactly as the intersection of the outer approximations $\{\mathcal C_\ell(d,n)\}_{\ell\ge2}$, reducing the original convex simulation problem to a hierarchy of finite-dimensional SDP feasibility tests.

The classical simulability hierarchy also admits a dual formulation. At level $\ell$, introduce Hermitian matrix variables $W_x^{(\ell)}, \,S_\mu^{(\ell)} \in \mathbb{H}_d$ and $Z_{j,i,\mu}^{(\ell)}\in\mathbb{H}_{d^\ell}$, real variables $y_{\mathbf{i},\mu}^{(\ell)}\in\mathbb{R}$, and a scalar $t^{(\ell)}\in\mathbb{R}$, together with complex-valued multipliers associated with the swap constraints. The resulting dual SDP can be written as
\begin{equation}\label{equ:maindualsdp} 
	\begin{aligned}
		\min \quad &
		\sum_x \mathrm{Tr}(W_x^{(\ell)}\rho_x)-t^{(\ell)} \\
		\mathrm{s.t.}\quad &
		\sum_{x:\mu(x)=i} W_x^{(\ell)}
		+\sum_{j=1}^{\ell}\mathrm{Tr}_{\setminus j}(Z_{j,i,\mu}^{(\ell)}) - S_\mu^{(\ell)}
		\succeq 0,
		\quad \forall\,i,\mu,\\
		&
		\sum_{j=1}^{\ell} Z_{j,i_j,\mu}^{(\ell)}
		+y_{\mathbf{i},\mu}^{(\ell)}\,\mathbb{I}_{d^\ell}
		+\mathcal{D}_{\mathbf{i},\mu}^{(\ell)}
		\preceq 0,
		\quad \forall\,\mathbf{i},\mu,\\
		&
		\mathrm{Tr}(S_\mu^{(\ell)}) + \sum_{\mathbf{i}} y_{\mathbf{i},\mu}^{(\ell)}-t^{(\ell)} \ge 0,
		\quad \forall\,\mu ,
	\end{aligned}
\end{equation}
where $\mathrm{Tr}_{\setminus j}$ denotes the partial trace over all tensor factors except the $j$th one. The Hermitian operator $\mathcal{D}_{\mathbf{i},\mu}^{(\ell)}$ is induced by the complex-valued dual multipliers associated with the swap constraints. Its explicit form, together with the derivation of Eq.~\eqref{equ:maindualsdp}, is given in Appendix~\ref{appD}.

By weak duality, any level-$\ell$ dual feasible point defines an affine functional
\begin{equation}\label{equ:dualwitness} 
\mathcal{W}^{(\ell)}(\boldsymbol{\rho})
:=
\sum_x \mathrm{Tr}(W_x^{(\ell)}\rho_x)-t^{(\ell)},
\end{equation}
that is nonnegative on $\mathcal{C}_\ell(d,n)$. Therefore, if there exists a dual feasible point such that $\mathcal{W}^{(\ell)}(\boldsymbol{\rho})<0$, then $\boldsymbol{\rho}\notin\mathcal{C}_\ell(d,n)$, and hence $\boldsymbol{\rho}\notin\mathcal{C}(d,n)$. Thus the operators $\{W_x^{(\ell)}\}_x$, together with the threshold $t^{(\ell)}$, define an explicit affine witness. Since evaluating this witness only requires the expectation values $\mathrm{Tr}(W_x^{(\ell)}\rho_x)$, it provides an experimentally accessible certificate for identifying state sets that cannot be classically simulated, without requiring full state-set tomography~\cite{1996horodeckiSeparabilityMixed,2007eisertQuantitativeEntanglement,2009guhneEntanglementDetection}.

\section{Numerical Results: Critical Visibilities and Witnesses}
\label{sec:numerics}

The SDP hierarchy immediately yields upper bounds on the critical visibility by replacing the exact set $\mathcal C(d,n)$ in Eq.~\eqref{equ:classical_visibility} with its level-$\ell$ relaxation, leading to
\[
\chi_\ell(\boldsymbol{\rho})
=
\max\left\{
v\in[0,1]\; \middle|\;
\{\mathrm{\Lambda}_v(\rho_x)\}_x\in\mathcal C_\ell(d,n)
\right\}.
\]
Since $\mathcal C_\ell(d,n)$ is an outer approximation to $\mathcal C(d,n)$, one has $\chi_\ell(\boldsymbol{\rho})\ge\chi(\boldsymbol{\rho})$. Moreover, the hierarchy structure implies $\chi_2(\boldsymbol{\rho})\ge\chi_3(\boldsymbol{\rho})\ge\cdots\ge\chi(\boldsymbol{\rho})$, as illustrated in Fig.~\ref{fig1}. 

\begin{figure}
	\centering
	\includegraphics[width=8.5cm]{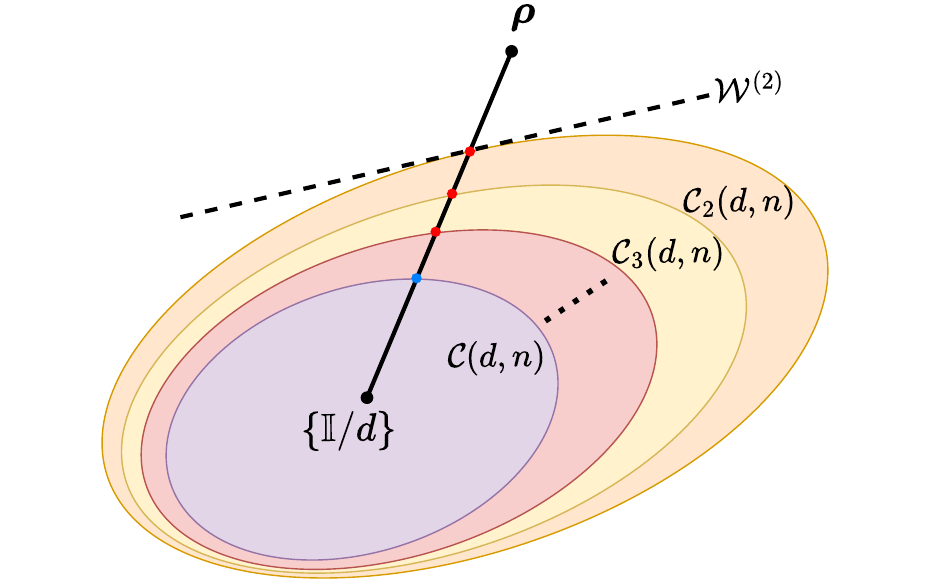}
	\caption{Geometric illustration of the hierarchy of convex relaxations. The black line represents the depolarizing trajectory $\{\mathrm{\Lambda}_v(\boldsymbol{\rho})\}_{v\in[0,1]}$, interpolating between the maximally mixed state set $\{\mathbb{I}/d\}$ and the target state set $\boldsymbol{\rho}$. The intersections of this trajectory with the boundaries of the hierarchy relaxations determine the visibility bounds $\chi_\ell$, shown as red points. As the hierarchy level increases, these bounds decrease monotonically and converge to the critical visibility $\chi$, shown in blue. The dashed line illustrates the separating hyperplane defined by an affine witness.}
	\label{fig1}
\end{figure}

As a first application, we investigate four symmetric qubit state sets, where $\mathrm\Psi_\phi:=|\phi\rangle\langle\phi|$. The BB84 state set is $\{\mathrm\Psi_0,\mathrm\Psi_1,\mathrm\Psi_+,\mathrm\Psi_-\}$, where $|\pm\rangle=(|0\rangle\pm|1\rangle)/\sqrt2$. The six-state set comprises the eigenstates of the three Pauli observables, namely $\{\mathrm\Psi_0,\mathrm\Psi_1,\mathrm\Psi_+,\mathrm\Psi_-,\mathrm\Psi_{+i},\mathrm\Psi_{-i}\}$, with $|\pm i\rangle=(|0\rangle\pm i|1\rangle)/\sqrt2$. The remaining examples are the trine state set $\{\mathrm\Psi_{\phi_k}\}_{k=0}^{2}$, where $|\phi_k\rangle=(|0\rangle+e^{2\pi i k/3}|1\rangle)/\sqrt2$, and the tetrahedral state set $\{\rho_a\}_{a=1}^{4}$, with $\rho_a=(\mathbb I+\mathbf r_a\cdot\boldsymbol\sigma)/2$ and $\mathbf r_a\in\{(1,1,1),(1,-1,-1),(-1,1,-1),(-1,-1,1)\}/\sqrt3$. Applying the level-2 classical simulability hierarchy to these four state sets under depolarizing noise, we obtain
\[
\chi_2
\approx
\frac{1}{\sqrt2},\quad
\frac{1}{\sqrt3},\quad
\frac{2}{3},\quad
\frac{1}{\sqrt3},
\]
respectively, with numerical errors below $10^{-10}$. The first value is consistent with the known analytical threshold for the noisy BB84 state set~\cite{2026cobucciOperationallyClassical}. For all four state sets, using the optimal level-$2$ solutions $\{c_\mu,\tau_{i,\mu}\}$ as guidance, we construct explicit classical simulations at the displayed visibility values. This establishes matching lower bounds and thereby identifies the displayed values as the critical visibilities. We further evaluate the hierarchy at levels $\ell=3$ and $\ell=4$, obtaining identical optima within numerical precision, indicating that low hierarchy levels already capture the relevant boundary of $\mathcal C(2,n)$ for these symmetric qubit state sets. These computations are performed on a PC equipped with an AMD Ryzen 9 7950X CPU and 32 GB of RAM, using MOSEK~\cite{mosek} as the SDP solver, and are completed within a few seconds. The data supporting the above analysis are provided in Appendix~\ref{appE}.

We next consider several higher-dimensional state sets. Specifically, $\boldsymbol{\rho}_1$, $\boldsymbol{\rho}_2$, and $\boldsymbol{\rho}_3$ consist of all states from $2$, $3$, and $4$ mutually unbiased bases (MUBs) in dimension $3$, respectively, while $\boldsymbol{\rho}_4$ is the symmetric informationally complete (SIC) state set in dimension $3$. We also consider
$
\boldsymbol{\rho}_5, \,\boldsymbol{\rho}_6
=
\{\ket{1}\bra{1},\ldots,\ket{d-1}\bra{d-1},\ket{u_d}\bra{u_d}\},
$
for $d=4$ and $5$, respectively, where
$
\ket{u_d}=d^{-1/2}\sum_{j=1}^{d}\ket{j}
$
is the maximally coherent state. The numerical results are summarized in Table~\ref{tab:highdim}. For each state set, the classical simulability hierarchy and the aggregated hierarchy provide level-$2$ upper bounds on $\chi$, whereas the method of Ref.~\cite{2026cobucciOperationallyClassical} provides lower bounds through explicit classical simulations over finite classical state sets.

\begin{table}[t]
	\centering
	\caption{
		Numerical bounds on the critical visibility for higher-dimensional state sets. Columns $\chi_2^{\mathrm{CSH}}$ and $\chi_2^{\mathrm{AH}}$ show the level-$2$ upper bounds from the classical simulability hierarchy and the aggregated hierarchy, respectively. Column $\kappa_{3000}$ shows the lower bound obtained by the method of Ref.~\cite{2026cobucciOperationallyClassical}, using $3000$ classical state sets. A dash indicates that the corresponding computation could not be completed within the available RAM. For the remaining computationally tractable instances, the computation time typically ranged from a few seconds to several hours, depending on the hierarchy level and the problem instance. The associated complexity estimates for these computations are given in Appendix~\ref{appC}.
	}
	\label{tab:highdim}
	\renewcommand{\arraystretch}{1.5}
	\begin{ruledtabular}
		\begin{tabular}{ccccc}
			State set 
			& $(d,n)$
			& \multicolumn{2}{c}{Upper bounds}
			& Lower bound \\
			\cline{3-4}
			& 
			& $\chi_2^{\mathrm{CSH}}$
			& $\chi_2^{\mathrm{AH}}$
			& $\kappa_{3000}$ \\
			\hline
			$\boldsymbol{\rho}_1$  
			& $(3,6)$ & 0.6667 & 0.6667 & 0.6122 \\
			$\boldsymbol{\rho}_2$  
			& $(3,9)$ & 0.5686 & 0.5686 & 0.5257 \\
			$\boldsymbol{\rho}_3$ 
			& $(3,12)$ & -- & 0.4729 & 0.4567 \\
			$\boldsymbol{\rho}_4$   
			& $(3,9)$ & 0.4698 & 0.4698 & 0.4516 \\
			$\boldsymbol{\rho}_5$  
			& $(4,4)$ & 0.6852 & 0.6852 & 0.5856 \\
			$\boldsymbol{\rho}_6$  
			& $(5,5)$ & 0.6784 & 0.6784 & 0.4891 
		\end{tabular}
	\end{ruledtabular}
\end{table}

The numerical results exhibit several notable features. First, whenever both level-$2$ upper bounds were computed, the classical simulability hierarchy and the aggregated hierarchy yielded identical values up to numerical precision. This suggests that, for these symmetric state sets, the additional hidden-variable structure retained in the classical simulability hierarchy may not improve the level-$2$ relaxation. Second, for the qutrit MUB state sets $\boldsymbol{\rho}_1$, $\boldsymbol{\rho}_2$, and $\boldsymbol{\rho}_3$, the level-$2$ upper bounds decrease as the number of MUBs increases, suggesting that the state sets become more robust against depolarizing noise as more mutually unbiased directions are included. Notably, these upper bounds coincide with those obtained using the witness method of Ref.~\cite{2026cobucciOperationallyClassical}. Third, the level-$2$ upper bound for the qutrit SIC state set $\boldsymbol{\rho}_4$ is close to that for $\boldsymbol{\rho}_3$, despite their different cardinalities and symmetry structures. This indicates that the critical visibility depends not only on the cardinality of the state set but also on its geometrical structure. Finally, the level-$2$ upper bounds for $\boldsymbol{\rho}_5$ and $\boldsymbol{\rho}_6$ are comparatively large and close to each other, consistent with their sparse structure consisting of one maximally coherent state together with $d-1$ computational-basis states.

For the qutrit state set $\boldsymbol{\rho}_1$, we also evaluated the classical simulability hierarchy at level $\ell=3$. The optimal value remained unchanged within numerical precision, giving $\chi_3(\boldsymbol{\rho}_1)\simeq2/3$. Moreover, as in the qubit examples, guided by the optimal numerical solution, we constructed an explicit classical simulation at visibility $v=2/3$. Hence, the level-2 and level-3 upper bounds are matched by explicit classical simulations, identifying $\chi(\boldsymbol{\rho}_1)=2/3$. This provides a higher-dimensional example in which a low hierarchy level already captures the relevant boundary of $\mathcal{C}(3,6)$. For the aforementioned four qubit state sets and for $\boldsymbol{\rho}_1$, we further extracted dual witness operators of the form defined in Eq.~\eqref{equ:dualwitness}. In these instances, the corresponding explicit classical simulations and the stability across hierarchy levels indicate that the resulting witnesses provide an informative characterization of the boundary of $\mathcal{C}(d,n)$. They therefore provide effective affine witnesses for detecting state sets that cannot be classically simulated. The explicit classical simulations and witness operators are presented in Appendix~\ref{appE}; additional implementation details can be found in Ref.~\cite{Li2026ClassicalSimulability}.

\section{Conclusion}
\label{sec:conclusion}

We characterized the set of classically simulable state sets through a complete SDP hierarchy. The construction of this hierarchy relies on a structural connection between the classical simulability of state sets and the projective simulability of auxiliary POVMs, with the latter serving as an intermediate object in the characterization. The associated dual programs also provided affine witnesses for detecting state sets that cannot be classically simulated. When applied to quantum state sets under depolarizing noise, the hierarchy yielded computable upper bounds on the critical visibility. In concurrent work, Cobucci \textit{et al.}~\cite{2026cobucciCertifyingCoherence} also developed both a complete SDP hierarchy and a computationally efficient SDP hierarchy for the same classical simulability problem. Their complete hierarchy is constructed from higher-order tensor-product extensions of the underlying hidden classical states, thereby providing a different structural route to characterizing classical simulability of state sets.

In our numerical investigation, the classical simulability hierarchy and its aggregated variant yielded identical upper bounds within numerical precision for the higher-dimensional state sets considered. Identifying conditions under which these two hierarchies yield the same numerical results remains an interesting direction for future research. More broadly, our results establish a systematic framework for investigating the classical simulability of quantum state sets and may also be useful for related convex-geometric problems, including quantum measurements without superposition~\cite{2026cobucciSimulatingQuantum} and the absolute dimensionality of quantum state sets~\cite{2024bernalAbsoluteDimensionality,2026dalessandroSemidefiniteBlockmatrix}.

\section*{Acknowledgments}

This work is supported by the National Natural Science Foundation of China (Grants No. 62571060, No. 62171056, and No. 62220106012).

\section*{Data availability}

The data that support the findings of this article are openly available~\cite{Li2026ClassicalSimulability}.

%

%\bibliography{Ref}

\appendix

\onecolumngrid

\newpage

\begin{appendix}

\setcounter{equation}{0}
\renewcommand\theequation{A\arabic{equation}}

\section{Rank-One Projective Simulability of the Auxiliary POVMs}
\label{appA}

In this Appendix, we justify the claim made after Eq.~\eqref{equ:classical2} that, for every deterministic assignment $\mu$ satisfying $c_\mu>0$, the normalized operators
$\{\tau_{i|\mu}\}_{i=1}^{d}$
define a POVM simulable by rank-one projective measurements. 

Starting from Eq.~\eqref{equ:classical2},
\begin{equation}\label{app1}
\rho_x
=
\sum_{\mu=1}^{d^n}
\sum_{i=1}^{d}
D(i|x,\mu)\tau_{i,\mu},
\end{equation}
where
\[
\tau_{i,\mu}
:=
\int d\lambda\,
p(\lambda)
q(\mu|\lambda)
|e_i^\lambda\rangle
\langle e_i^\lambda|.
\]
By construction,
$\tau_{i,\mu}\succeq0$.
Moreover,
$
\mathrm{Tr}(\tau_{i,\mu})
=
\int d\lambda\,
p(\lambda)
q(\mu|\lambda)
=:c_\mu,
$
which is independent of $i$. Since $q(\mu|\lambda)$ is a conditional probability distribution and $p(\lambda)$ is normalized, it follows that $c_\mu\ge0$ and $\sum_\mu c_\mu=1$.

If $c_\mu=0$, positivity implies
$\tau_{i,\mu}=0$
for every $i$, and such terms contribute trivially to Eq.~\eqref{app1}. We therefore restrict attention to deterministic assignments satisfying
$c_\mu>0$
and define
$
\tau_{i,\mu}
=
c_\mu\tau_{i|\mu}
$
with
$
\tau_{i|\mu}
=
\tau_{i,\mu}/c_\mu.
$
Summing over $i$, we find
\[
\sum_{i=1}^{d}
\tau_{i|\mu}
=
\frac{1}{c_\mu}
\int d\lambda\,
p(\lambda)
q(\mu|\lambda)
\sum_{i=1}^{d}
|e_i^\lambda\rangle\langle e_i^\lambda|
=
\mathbb{I}_d,
\]
showing that
$\{\tau_{i|\mu}\}_{i=1}^{d}$
forms a $d$-outcome POVM.

Furthermore,
\begin{equation}\label{equ:normalizedoperator}
\tau_{i|\mu}
=
\int d\lambda\,
\frac{
	p(\lambda)q(\mu|\lambda)
}{
	c_\mu
}
|e_i^\lambda\rangle\langle e_i^\lambda|.
\end{equation}
For convenience, define
$
f(\lambda|\mu)
:=
{p(\lambda)q(\mu|\lambda)}/{c_\mu},
$
which is a normalized probability distribution over $\lambda$. Then
$
\tau_{i|\mu}
=
\int d\lambda\,
f(\lambda|\mu)
|e_i^\lambda\rangle\langle e_i^\lambda|.
$ For each fixed $\lambda$, the operators
$\{
|e_i^\lambda\rangle\langle e_i^\lambda|
\}_{i=1}^{d}$
form a rank-one projective measurement, while
$f(\lambda|\mu)$
defines a probability distribution over $\lambda$. Therefore,
$\{\tau_{i|\mu}\}_{i=1}^{d}$
is a convex mixture of rank-one projective measurements and hence is simulable by rank-one projective measurements~\cite{2017oszmaniecSimulatingPositiveoperatorvalued}.

\setcounter{equation}{0}
\renewcommand\theequation{B\arabic{equation}}
\section{Proof of Completeness of the Rank-One Projective-Simulability Hierarchy}
\label{appB}

This appendix contains the proof of Theorem~\hyperlink{theorem1}{1}, which proceeds in two steps. We first consider the larger set $\mathcal P(d,m)$ consisting of POVMs simulable by arbitrary projective measurements, without imposing the rank-one restriction. For the hierarchy introduced in Ref.~\cite{2026brinsterRobustCertification}, we prove that the corresponding semidefinite relaxations ${\mathcal P_\ell(d,m)}_{\ell\ge2}$ are asymptotically complete. We then refine the hierarchy to the rank-one setting considered in the main text and prove Theorem~\hyperlink{theorem1}{1} by incorporating the trace-one constraints.

\subsection{Completeness of the Projective-Simulability Hierarchy}
\label{subappB}

We begin by recalling the larger set $\mathcal P(d,m)$ consisting of POVMs with $m$ outcomes acting on a $d$-dimensional Hilbert space that are simulable by arbitrary projective measurements~\cite{2017oszmaniecSimulatingPositiveoperatorvalued,2025khandelwalSimulatingQuantum,2026cobucciMaximallyNonprojective}. A POVM
$\mathbf M=\{M_i\}_{i=1}^{m}$
belongs to
$\mathcal P(d,m)$
if there exists a probability distribution
$\{w_k\}_k$
and projective measurements
$\{P_i^k\}_{i=1}^{m}$
such that
$
M_i=\sum_k w_k P_i^k
$
for all
$i$.
To characterize this set, we consider the hierarchy of semidefinite relaxations
$\{\mathcal P_\ell(d,m)\}_{\ell\ge2}$
introduced in Ref.~\cite{2026brinsterRobustCertification}:
\begin{equation}
	\label{equ:lift}
	\left\{
	\begin{aligned}
		&\mathrm{\Pi}^{(\ell)}_{\mathbf{i}} \succeq 0,
		\quad \forall\, \mathbf{i}\in[m]^\ell,\\[1mm]
		&\sum_{\mathbf{i}\setminus i_j}\mathrm{\Pi}^{(\ell)}_{\mathbf{i}}
		=\mathbb{I}_{d^{j-1}}\otimes M_{i_j}\otimes\mathbb{I}_{d^{\ell-j}},
		\quad \forall\, j\in[\ell],\ i_j\in[m],\\[1mm]
		&\mathrm{Tr}_\alpha\!\left(
		V_{(\alpha\beta)}\mathrm{\Pi}^{(\ell)}_{\mathbf{i}}
		\right)= 0,
		\quad \forall\, \mathbf{i},\alpha<\beta:i_\alpha\neq i_\beta .
	\end{aligned}
	\right.
\end{equation}
Our contribution here is to prove the asymptotic completeness of this hierarchy.

%Our contribution here is to prove the asymptotic completeness of this hierarchy, as stated in Lemma~\ref{lem:hierarchy1}. Several technical lemmas required for its proof are deferred to Lemmas~\ref{lem:symmetrization}--\ref{lem:definetti-normalization}, which appear subsequently to preserve the flow of the main argument.
\begin{lemma}
	\label{lem:hierarchy1}
	Let $\mathbf M=\{M_i\}_{i=1}^{m}$ be a POVM on a $d$-dimensional Hilbert space. Then
	$
	\mathbf M\in\mathcal P_\ell(d,m)
	$
	for every
	$
	\ell\ge2
	$
	if and only if
	$
	\mathbf M\in\mathcal P(d,m).
	$
\end{lemma}

\begin{proof}[Proof of Lemma~\ref{lem:hierarchy1}]
	
	The ``if'' direction is immediate. Suppose that $\mathbf{M}=\{M_i\}_{i=1}^m \in \mathcal{P}(d,m)$. By definition, there exist weights $w_k\ge 0$ with $\sum_k w_k=1$, and projective measurements $\mathbf{P}^{\,k}=\{P_i^k\}_{i=1}^m$, such that
	\begin{equation}\label{}
	M_i=\sum_k w_k P_i^k,\quad \forall\, i\in[m].
	\end{equation}
	For each $\ell$, define
	\begin{equation}\label{}
	\mathrm{\Pi}_{\mathbf{i}}^{(\ell)}
	:=
	\sum_k w_k \, P_{i_1}^k\otimes P_{i_2}^k\otimes \cdots \otimes P_{i_\ell}^k,
	\quad \mathbf{i}=(i_1,\dots,i_\ell)\in[m]^\ell.
	\end{equation}
	Clearly, $\mathrm{\Pi}_{\mathbf{i}}^{(\ell)}\succeq 0$. Moreover,
	\begin{equation}\label{}
	\sum_{\mathbf{i}\setminus i_j}\mathrm{\Pi}_{\mathbf{i}}^{(\ell)}
	=
	\sum_k w_k
	\left(\sum_{i_1} P_{i_1}^k\right)\otimes \cdots \otimes P_{i_j}^k \otimes \cdots \otimes \left(\sum_{i_\ell} P_{i_\ell}^k\right)
	=
	\mathbb{I}_{d^{j-1}}\otimes M_{i_j}\otimes \mathbb{I}_{d^{\ell-j}},
	\end{equation}
	so the marginal constraint is satisfied. Finally, if $i_\alpha\neq i_\beta$, then
	\begin{equation}\label{}
	\begin{aligned}
		\mathrm{Tr}_\alpha\!\left(V_{(\alpha\beta)}\mathrm{\Pi}_{\mathbf{i}}^{(\ell)}\right)
		&=
		\sum_k w_k\,
		\mathrm{Tr}_\alpha\!\left[
		V_{(\alpha\beta)}
		\left(
		P_{i_\alpha}^k\otimes P_{i_\beta}^k
		\right)
		\right]
		\otimes
		\bigotimes_{\gamma\neq \alpha,\beta} P_{i_\gamma}^k \\
		&=
		\sum_k w_k\,
		\left(P_{i_\alpha}^k P_{i_\beta}^k\right)
		\otimes
		\bigotimes_{\gamma\neq \alpha,\beta} P_{i_\gamma}^k
		=
		0,
	\end{aligned}
	\end{equation}
	because the effects of a projective measurement are pairwise orthogonal. Hence $\mathbf{M}\in\mathcal{P}_\ell(d,m)$ for every $\ell\geq 2$.
	
	We now prove the converse direction. Assume that $\mathbf{M}\in\mathcal{P}_\ell(d,m)$ for every $\ell\geq 2$. For each $\ell$, choose a feasible set $\{\mathrm{\Pi}_{\mathbf{i}}^{(\ell)}\}_{\mathbf{i}\in[m]^\ell}$, and define the unnormalized operator
	\begin{equation}\label{}
	\mathrm{\Omega}^{(\ell)}
	:=
	\sum_{\mathbf{i}\in[m]^\ell}
	|\mathbf{i}\rangle\langle \mathbf{i}|_C \otimes \mathrm{\Pi}_{\mathbf{i}}^{(\ell)},
	\end{equation}
	where $C=C_1\otimes\cdots\otimes C_\ell$ is a classical register of dimension $m^\ell$, and $Q=Q_1\otimes\cdots\otimes Q_\ell$ is a quantum register of dimension $d^\ell$. By the marginal constraint 
	$
	\sum_{\mathbf{i}\in[m]^\ell}\mathrm{\Pi}_{\mathbf{i}}^{(\ell)}=\mathbb{I}_{d^\ell},
	$
	and therefore
	\begin{equation}\label{}
	\mathrm{Tr}\!\left(\mathrm{\Omega}^{(\ell)}\right)
	=
	\sum_{\mathbf{i}\in[m]^\ell}\mathrm{Tr}\!\left(\mathrm{\Pi}_{\mathbf{i}}^{(\ell)}\right)
	=
	d^\ell.
	\end{equation}
	Hence, we define the normalized $\ell$-partite CQ density operator as
	\begin{equation}\label{}
	\varrho^{(\ell)}
	:=
	\frac{1}{d^\ell}\mathrm{\Omega}^{(\ell)}
	=
	\frac{1}{d^\ell}
	\sum_{\mathbf{i}\in[m]^\ell}
	|\mathbf{i}\rangle\langle \mathbf{i}|_C \otimes \mathrm{\Pi}_{\mathbf{i}}^{(\ell)}.
	\end{equation}
	
	By Lemma~\ref{lem:symmetrization}, we may replace $\{\mathrm{\Pi}_{\mathbf{i}}^{(\ell)}\}$ by its symmetrized version without affecting feasibility. Thus, without loss of generality, we assume that each $\varrho^{(\ell)}$ is permutation invariant. Since $\varrho^{(\ell)}$ is an $\ell$-partite CQ state on the global space $(C\otimes Q)^{\otimes \ell}$, the sequence $\{\varrho^{(\ell)}\}_{\ell\ge 2}$ becomes infinitely exchangeable in the limit $\ell\to\infty$. The quantum de Finetti theorem \cite{2002cavesUnknownQuantum} then implies that, for any fixed positive integer $k$ with $k\ll \ell$, if we denote by $\varrho_k^{(\ell)}:=\mathrm{Tr}_{k+1,\dots,\ell}(\varrho^{(\ell)})$ the reduced state of $\varrho^{(\ell)}$ on the first $k$ subsystems $(C_1\otimes Q_1)\otimes\cdots\otimes(C_k\otimes Q_k)$, then
	\begin{equation}\label{quantumdeFinetti}
	\lim_{\ell\to\infty}\varrho_k^{(\ell)}=\int_{\mathcal{D}(CQ)} d\eta(\varsigma)\,\varsigma^{\otimes k},
	\end{equation}
	where $\varsigma\in\mathcal{D}(CQ)$ is a density operator on the single-copy space $C_1\otimes Q_1$, and $\eta$ is the corresponding probability measure. Moreover, for every finite $\ell$, the global state $\varrho^{(\ell)}$ is strictly diagonal on the classical registers $C^{\otimes \ell}$. Since neither taking reduced states nor the asymptotic limit changes this CQ structure, the measure $\eta$ must be supported entirely on single-copy CQ states. Accordingly, $\varsigma$ necessarily takes the form
	\begin{equation}\label{}
	\varsigma=\sum_{x=1}^m |x\rangle\langle x|\otimes \varsigma_x,
	\end{equation}
	where each $\varsigma_x$ is a subnormalized positive semidefinite operator on $Q_1$, i.e., $\varsigma_x\succeq 0$, and satisfies $\mathrm{Tr}(\sum_{x=1}^m \varsigma_x)=1$. 
			
	We first identify the one-copy marginal. From the definition of $\varrho^{(\ell)}$ and the marginal constraint,
	\begin{equation}\label{onecopylhs}
	\begin{aligned}
		\varrho_1^{(\ell)}
		&=
		\mathrm{Tr}_{C_2Q_2\cdots C_\ell Q_\ell}\!\left(
		\frac{1}{d^\ell}
		\sum_{\mathbf{i}}
		|\mathbf{i}\rangle\langle \mathbf{i}|_C \otimes \mathrm{\Pi}_{\mathbf{i}}^{(\ell)}
		\right) \\
		&=
		\frac{1}{d^\ell}
		\sum_{i_1=1}^m
		|i_1\rangle\langle i_1|
		\otimes
		\mathrm{Tr}_{Q_2\cdots Q_\ell}
		\left(
		\sum_{\mathbf{i}\setminus i_1}\mathrm{\Pi}_{\mathbf{i}}^{(\ell)}
		\right) \\
		&=
		\frac{1}{d^\ell}
		\sum_{x=1}^m
		|x\rangle\langle x|\otimes
		\mathrm{Tr}_{Q_2\cdots Q_\ell}
		\left(
		M_x\otimes \mathbb{I}_{d^{\ell-1}}
		\right) \\
		&=
		\frac{1}{d}
		\sum_{x=1}^m |x\rangle\langle x|\otimes M_x.
	\end{aligned}
	\end{equation}
	On the other hand, the quantum de Finetti representation in Eq.~\eqref{quantumdeFinetti} with $k=1$ gives
	\begin{equation}\label{onecopyrhs}
	\lim_{\ell\to\infty}\varrho_1^{(\ell)}
	=
	\int d\eta(\varsigma)\,\varsigma
	=
	\sum_{x=1}^m |x\rangle\langle x|\otimes \left(\int d\eta(\varsigma)\,\varsigma_x\right).
	\end{equation}
	Comparing Eqs.~\eqref{onecopylhs} and \eqref{onecopyrhs}, we obtain
	\begin{equation}\label{}
	M_x
	=
	d\int d\eta(\varsigma)\,\varsigma_x,\quad \forall\, x \in[m].
	\end{equation}
	Define
	$
	N_x:=d\,\varsigma_x.
	$
	By Lemma~\ref{lem:definetti-normalization}, for $\eta$-almost everywhere (a.e.), the set $\{N_x\}_{x=1}^m$ forms a valid $m$-outcome POVM, i.e.,
	\begin{equation}\label{}
	N_x\succeq 0,
	\quad
	\sum_{x=1}^m N_x=\mathbb{I}_d.
	\end{equation}
	This shows that the collection $N=\{N_x\}_{x=1}^m$ indeed forms a valid $m$-outcome POVM. We may therefore equivalently rewrite the measure $\eta(\varsigma)$ as a measure $\eta(N)$ supported on the set $\mathbb{M}(d,m)$ of all valid POVMs. The target POVM $\mathbf{M}$ can then be expressed as a convex mixture of single-copy POVMs
	\begin{equation}\label{}
	M_x=\int_{\mathbb{M}(d,m)} d\eta(N)\,N_x,\quad \forall\, x \in[m].
	\end{equation}
	
	At this point, by setting $k=2$, we can write the two-body limiting reduced state explicitly as
	\begin{equation}\label{twobodylimit}
	\lim_{\ell\to\infty}\varrho_2^{(\ell)}
	=
	\frac{1}{d^2}\int_{\mathbb{M}(d,m)} d\eta(N)\sum_{u,v=1}^m
	\left(|u\rangle\langle u|\otimes |v\rangle\langle v|\right)\otimes (N_u\otimes N_v),
	\end{equation}
	which provides the starting point for eliminating the cross terms via the swap constraint.	To exploit the vanishing swap constraint in the SDP hierarchy, we introduce a global observable on the bipartite space $(C_1\otimes Q_1)\otimes (C_2\otimes Q_2)$. For any two distinct outcome labels $u\neq v\in[m]$, define
	$
	O_{u,v}:=\left(|u\rangle\langle u|_{C_1}\otimes |v\rangle\langle v|_{C_2}\right)\otimes V_{Q_1Q_2},
	$
	where $V_{Q_1Q_2}$ is the swap operator acting on the two quantum copies. We now compute its expectation value on the two-body reduced state $\varrho_2^{(\ell)}$ for finite $\ell\ge 2$. From the definition of the global CQ state, tracing out the last $\ell-2$ subsystems gives
	\begin{equation}\label{}
	\varrho_2^{(\ell)}
	=
	\frac{1}{d^\ell}\sum_{i_1,i_2=1}^m
	\left(|i_1\rangle\langle i_1|\otimes |i_2\rangle\langle i_2|\right)\otimes
	\mathrm{Tr}_{Q_3\cdots Q_\ell}
	\left(\sum_{\mathbf{i}\setminus i_1,i_2} {\mathrm{\Pi}}_{\mathbf{i}}^{(\ell)}\right).
	\end{equation}
	Since the classical basis $\{|i\rangle\}$ is orthogonal, in evaluating $\mathrm{Tr}(O_{u,v}\varrho_2^{(\ell)})$ only the terms with $i_1=u$ and $i_2=v$ contribute. Therefore,
	\begin{equation}\label{}
	\begin{aligned}
		\mathrm{Tr}(O_{u,v}\varrho_2^{(\ell)})
		&=
		\frac{1}{d^\ell}
		\mathrm{Tr}_{Q_1Q_2}\!\left[
		V_{Q_1Q_2}\,
		\mathrm{Tr}_{Q_3\cdots Q_\ell}
		\left(\sum_{\mathbf{i}\setminus i_1=u,\; i_2=v} {\mathrm{\Pi}}_{\mathbf{i}}^{(\ell)}\right)
		\right] \\
		&=
		\frac{1}{d^\ell}
		\sum_{\mathbf{i}\setminus i_1=u,\; i_2=v}
		\mathrm{Tr}_{Q_1\cdots Q_\ell}\!\left[
		\left(V_{Q_1Q_2}\otimes \mathbb{I}_{Q_3\cdots Q_\ell}\right)
		 {\mathrm{\Pi}}_{\mathbf{i}}^{(\ell)}
		\right].
	\end{aligned}
	\end{equation}
	Because $u\neq v$, every tensor $ {\mathrm{\Pi}}_{\mathbf{i}}^{(\ell)}$ appearing in the sum satisfies $i_1\neq i_2$. By the defining constraint of $\mathcal{P}_\ell(d,m)$, for every such $\mathbf{i}$ we have $\mathrm{Tr}_{Q_1}\!\left(V_{Q_1Q_2} {\mathrm{\Pi}}_{\mathbf{i}}^{(\ell)}\right)=0$. Hence its full trace also vanishes, and we conclude that for every finite $\ell$ and every $u\neq v$,
	\begin{equation}\label{}
	\mathrm{Tr}(O_{u,v}\varrho_2^{(\ell)})=0.
	\end{equation}
	
	We now take the asymptotic limit $\ell\to\infty$. Since the trace is a continuous linear functional, the limit of the expectation values equals the expectation value on the limiting state:
	\begin{equation}\label{}
	\lim_{\ell\to\infty}\mathrm{Tr}(O_{u,v}\varrho_2^{(\ell)})
	=
	\mathrm{Tr}\!\left(
	O_{u,v}\left[\lim_{\ell\to\infty}\varrho_2^{(\ell)}\right]
	\right)
	=
	0.
	\end{equation}
	Substituting the two-body limiting state derived from the quantum de Finetti theorem, namely Eq.~\eqref{twobodylimit}, into the above expression, and using the orthogonality of the classical basis so that only the term $u'=u$ and $v'=v$ survives, we obtain
	\begin{equation}\label{}
	\frac{1}{d^2}\int_{\mathbb{M}(d,m)} d\eta(N)\,
	\mathrm{Tr}_{Q_1Q_2}\!\left[
	V_{Q_1Q_2}(N_u\otimes N_v)
	\right]
	=
	0.
	\end{equation}
	Using the identity $\mathrm{Tr}_{AB}[V_{AB}(A\otimes B)]=\mathrm{Tr}(AB)$, this reduces to
	\begin{equation}\label{}
	\int_{\mathbb{M}(d,m)} d\eta(N)\,\mathrm{Tr}(N_uN_v)=0,
	\quad \forall\,u\neq v.
	\end{equation}
	
	Now $\mathbb{M}(d,m)$ is precisely the set of all valid $m$-outcome POVMs. Hence, for every $N\in\mathbb{M}(d,m)$, the effects are positive semidefinite: $N_u\succeq 0$ and $N_v\succeq 0$. It follows that $\mathrm{Tr}(N_uN_v)\ge 0$. Since the integrand is nonnegative everywhere and its integral vanishes, basic measure theory implies that
	\begin{equation}\label{}
	\mathrm{Tr}(N_uN_v)=0
	\quad
	(\eta\text{-a.e.}),\quad \forall\,u\neq v.
	\end{equation}
	For positive semidefinite operators, vanishing overlap trace is equivalent to orthogonality, and therefore
	\begin{equation}\label{}
	N_uN_v=0
	\quad
	(\eta\text{-a.e.}),\quad \forall\,u\neq v.
	\end{equation}
	Using the POVM completeness relation $\sum_{x=1}^m N_x=\mathbb{I}_d$, we further obtain, for each $k$, $N_k=N_k\mathbb{I}_d=N_k(\sum_{x=1}^m N_x)=N_k^2+\sum_{x\neq k}N_kN_x$. Since all cross terms vanish $\eta$-a.e., it follows that
	\begin{equation}\label{}
	N_k=N_k^2
	\quad
	(\eta\text{-a.e.}).
	\end{equation}
	An operator that is both positive semidefinite and idempotent must be an orthogonal projector. Therefore, in the limit $\ell\to\infty$, the de Finetti measure $\eta(N)$ is supported entirely on projective measurements.
	
	Finally, since
	\begin{equation}\label{equ:simul}
	M_x=\int_{\mathbb{M}(d,m)} d\eta(N)\,N_x,\quad \forall\, x \in [d],
	\end{equation}
	and the measure $\eta$ is supported on projective measurements, the target POVM $\{M_x\}$ is exactly a convex mixture of projective measurements. Hence $\mathbf{M}\in\mathcal{P}(d,m)$.

\end{proof}

\begin{lemma}
	\label{lem:symmetrization}
	Let $S_\ell$ be the symmetric group on $\ell$ elements. Given any feasible solution $\{\mathrm{\Pi}_{\mathbf{i}}\}_{\mathbf{i}}$ for $\mathcal{P}_\ell(d,m)$, define its symmetrization by group averaging over $S_\ell$:
	\begin{equation}\label{}
	\widetilde{\mathrm{\Pi}}_{\mathbf{i}}:=\frac{1}{\ell!}\sum_{\pi\in S_\ell} W_\pi \mathrm{\Pi}_{\pi^{-1}(\mathbf{i})} W_\pi^\dagger.
	\end{equation}
	Then $\{\widetilde{\mathrm{\Pi}}_{\mathbf{i}}\}_{\mathbf{i}}$ remains a feasible solution to $\mathcal{P}_\ell(d,m)$, i.e., it satisfies the positivity, marginal, and swap constraints listed in Eq.~\eqref{equ:lift}. Moreover, the CQ state constructed from $\{\widetilde{\mathrm{\Pi}}_{\mathbf{i}}\}_{\mathbf{i}}$ is exactly permutation invariant, and hence exchangeable.
\end{lemma}

\begin{proof}[Proof of Lemma~\ref{lem:symmetrization}]
	For any permutation $\pi\in S_\ell$, let $\pi(\mathbf{i})$ denote the permuted outcome string whose $k$th component is $(\pi(\mathbf{i}))_k=i_{\pi^{-1}(k)}$. Let $W_\pi$ be the natural unitary representation of $S_\ell$ on $\mathcal{H}^{\otimes \ell}$, characterized by
	\begin{equation}\label{} 
	W_\pi (A_1\otimes\cdots\otimes A_\ell) W_\pi^\dagger
	=
	A_{\pi^{-1}(1)}\otimes\cdots\otimes A_{\pi^{-1}(\ell)}.
	\end{equation}
	We verify in turn that $\{\widetilde{\mathrm{\Pi}}_{\mathbf{i}}\}_{\mathbf{i}}$ is still a valid solution to $\mathcal{P}_\ell(d,m)$.
	
	First, positivity is immediate. Since each $\mathrm{\Pi}_{\pi^{-1}(\mathbf{i})}\succeq 0$ and $W_\pi$ is unitary, we have
	\begin{equation}\label{} 
		W_\pi \mathrm{\Pi}_{\pi^{-1}(\mathbf{i})} W_\pi^\dagger \succeq 0.
	\end{equation}
	Hence $\widetilde{\mathrm{\Pi}}_{\mathbf{i}}\succeq 0$ as a finite average of positive semidefinite operators.
	
	Next, we prove the marginal condition. For convenience, write 
	$
	M_x^{(j)}:=\mathbb{I}_{d^{j-1}}\otimes M_x\otimes \mathbb{I}_{d^{\ell-j}}.
	$
	We must show that, for every position $j$ and every outcome $x$,
	\begin{equation}\label{} 
		\sum_{\mathbf{i} \setminus i_j=x}\widetilde{\mathrm{\Pi}}_{\mathbf{i}}=M_x^{(j)}.
	\end{equation}
	Substituting the definition of $\widetilde{\mathrm{\Pi}}_{\mathbf{i}}$ gives
	\begin{equation}\label{}
	\sum_{\mathbf{i}\setminus i_j=x}\widetilde{\mathrm{\Pi}}_{\mathbf{i}}
	=
	\frac{1}{\ell!}\sum_{\pi\in S_\ell}
	W_\pi
	\left(
	\sum_{\mathbf{i}\setminus i_j=x}\mathrm{\Pi}_{\pi^{-1}(\mathbf{i})}
	\right)
	W_\pi^\dagger.
	\end{equation}
	Now set $\mathbf{k}=\pi^{-1}(\mathbf{i})$. The condition $i_j=x$ is equivalent to $k_{\pi^{-1}(j)}=x$. Writing $p=\pi^{-1}(j)$, the marginal constraint for the original solution yields
	$
	\sum_{\mathbf{k}\setminus k_p=x}\mathrm{\Pi}_{\mathbf{k}}=M_x^{(p)}.
	$ 
	Therefore,
	\begin{equation}\label{}
	\sum_{\mathbf{i}\setminus i_j=x}\widetilde{\mathrm{\Pi}}_{\mathbf{i}}
	=
	\frac{1}{\ell!}\sum_{\pi\in S_\ell}
	W_\pi M_x^{(p)} W_\pi^\dagger.
	\end{equation}
	Since $p=\pi^{-1}(j)$, we have $\pi(p)=j$, and thus
	$
	W_\pi M_x^{(p)} W_\pi^\dagger=M_x^{(j)}.
	$
	Each term in the average is therefore exactly $M_x^{(j)}$, which proves the marginal condition.
	
	We now verify the swap constraint. Suppose that $i_\alpha\neq i_\beta$, then
	\begin{equation}\label{}
	\mathrm{Tr}_\alpha\!\left(V_{(\alpha\beta)}\widetilde{\mathrm{\Pi}}_{\mathbf{i}}\right)
	=
	\frac{1}{\ell!}\sum_{\pi\in S_\ell}
	\mathrm{Tr}_\alpha\!\left[
	V_{(\alpha\beta)}W_\pi \mathrm{\Pi}_{\pi^{-1}(\mathbf{i})} W_\pi^\dagger
	\right].
	\end{equation}
	Let $\mathbf{k}=\pi^{-1}(\mathbf{i})$, and define $u=\pi^{-1}(\alpha)$ and $v=\pi^{-1}(\beta)$. Since $i_\alpha\neq i_\beta$, it follows that $k_u\neq k_v$. Using  the algebraic relations of the permutation group, the swap operator $V_{(uv)}$ and the global permutation unitary $W_\pi$ satisfy the identity
	\begin{equation}\label{}
	W_\pi V_{(u v)} W_\pi^\dagger = V_{(\alpha\beta)},
	\end{equation}
	or equivalently,
	\begin{equation}\label{}
	V_{(\alpha\beta)}W_\pi = W_\pi V_{(u v)},
	\end{equation}
	we obtain
	\begin{equation}\label{}
	\mathrm{Tr}_\alpha\!\left[
	V_{(\alpha\beta)}W_\pi \mathrm{\Pi}_{\mathbf{k}} W_\pi^\dagger
	\right]
	=
	\mathrm{Tr}_\alpha\!\left[
	W_\pi (V_{(u v)}\mathrm{\Pi}_{\mathbf{k}}) W_\pi^\dagger
	\right].
	\end{equation}
	By Lemma \ref{lem:partial-trace-covariance}, taking the partial trace over the $\alpha$th subsystem of a quantum state after a global permutation $\pi$ is exactly equivalent to first taking the partial trace over the $\pi^{-1}(\alpha)$th subsystem of the original state, namely the $u$th subsystem, and then applying the induced permutation on the remaining $\ell-1$ subsystems, i.e.,
	\begin{equation}\label{}
	\mathrm{Tr}_\alpha \left[ W_\pi (V_{(uv)} \mathrm{\Pi}_{\mathbf{k}}) W_\pi^\dagger \right] = \widetilde{W}_\pi \left[ \mathrm{Tr}_u (V_{(uv)} \mathrm{\Pi}_{\mathbf{k}}) \right] \widetilde{W}_\pi^\dagger,
	\end{equation}
	where $\widetilde{W}_\pi$ denotes the induced permutation unitary on the remaining $\ell-1$ tensor factors. Since $k_u\neq k_v$, the swap constraint for the original solution implies 
	$
	\mathrm{Tr}_u\!\left(V_{(u v)}\mathrm{\Pi}_{\mathbf{k}}\right)=0.
	$
	Hence every term in the above average vanishes, and therefore
	$
	\mathrm{Tr}_\alpha\!\left(V_{(\alpha\beta)}\widetilde{\mathrm{\Pi}}_{\mathbf{i}}\right)=0.
	$
	This proves the swap constraint.
	
	It remains to show permutation invariance of the associated CQ state. First observe that, for any $\varsigma\in S_\ell$,
	\begin{equation}\label{}
	W_\varsigma \widetilde{\mathrm{\Pi}}_{\mathbf{i}} W_\varsigma^\dagger
	=
	\frac{1}{\ell!}\sum_{\pi\in S_\ell}
	W_\varsigma W_\pi \mathrm{\Pi}_{\pi^{-1}(\mathbf{i})} W_\pi^\dagger W_\varsigma^\dagger
	=
	\frac{1}{\ell!}\sum_{\pi'\in S_\ell}
	W_{\pi'} \mathrm{\Pi}_{\pi'^{-1}(\varsigma(\mathbf{i}))} W_{\pi'}^\dagger
	=
	\widetilde{\mathrm{\Pi}}_{\varsigma(\mathbf{i})},
	\end{equation}
	where in the second equality we set $\pi'=\varsigma\pi$. By closure of the group, summing over all $\pi$ is equivalent to summing over all $\pi'$. We now construct the CQ state
	\begin{equation}\label{}
	\varrho^{(\ell)}:=\frac{1}{d^\ell}\sum_{\mathbf{i}\in[m]^\ell}
	|\mathbf{i}\rangle\langle \mathbf{i}|_C\otimes \widetilde{\mathrm{\Pi}}_{\mathbf{i}}.
	\end{equation}
	Let $\mathcal{U}_\varsigma:=P_\varsigma\otimes W_\varsigma$, where $P_\varsigma|\mathbf{i}\rangle=|\varsigma(\mathbf{i})\rangle$ permutes the classical register. Then
	\begin{equation}\label{}
	\begin{aligned}
		\mathcal{U}_\varsigma \varrho^{(\ell)} \mathcal{U}_\varsigma^\dagger
		&=
		\frac{1}{d^\ell}\sum_{\mathbf{i}}
		\left(P_\varsigma |\mathbf{i}\rangle\langle \mathbf{i}| P_\varsigma^\dagger\right)
		\otimes
		\left(W_\varsigma \widetilde{\mathrm{\Pi}}_{\mathbf{i}} W_\varsigma^\dagger\right)\\
		&=
		\frac{1}{d^\ell}\sum_{\mathbf{i}}
		|\varsigma(\mathbf{i})\rangle\langle \varsigma(\mathbf{i})|
		\otimes
		\widetilde{\mathrm{\Pi}}_{\varsigma(\mathbf{i})}.
	\end{aligned}
	\end{equation}
	Since $\varsigma$ acts bijectively on $[m]^\ell$, relabeling $\mathbf{j}=\varsigma(\mathbf{i})$ yields
	\begin{equation}\label{}
	\mathcal{U}_\varsigma \varrho^{(\ell)} \mathcal{U}_\varsigma^\dagger
	=
	\frac{1}{d^\ell}\sum_{\mathbf{j}}
	|\mathbf{j}\rangle\langle \mathbf{j}|
	\otimes
	\widetilde{\mathrm{\Pi}}_{\mathbf{j}}
	=
	\varrho^{(\ell)}.
	\end{equation}
	Hence $\varrho^{(\ell)}$ is exactly permutation invariant, i.e., exchangeable.
\end{proof}

\begin{lemma}
	\label{lem:partial-trace-covariance}
	Let $O$ be an arbitrary linear operator acting on $\mathcal{H}^{\otimes \ell}$, and let $W_\pi$ denote the unitary representation of a permutation $\pi\in S_\ell$. Then
	\begin{equation}\label{}
	\mathrm{Tr}_\alpha\!\left[ W_\pi O W_\pi^\dagger \right]
	=
	\widetilde{W}_\pi \,\mathrm{Tr}_u(O)\,\widetilde{W}_\pi^\dagger,
	\end{equation}
	where $u=\pi^{-1}(\alpha)$, and $\widetilde{W}_\pi$ is the induced permutation unitary acting on the remaining $\ell-1$ tensor factors.
\end{lemma}

\begin{proof}[Proof of Lemma~\ref{lem:partial-trace-covariance}]
	Since any operator on $\mathcal{H}^{\otimes \ell}$ can be written as a linear combination of product operators, and both partial trace and conjugation are linear, it suffices to prove the identity for
	\begin{equation}\label{}
	O=A_1\otimes A_2\otimes\cdots\otimes A_\ell.
	\end{equation}
	By definition of $W_\pi$, we have
	$
	W_\pi O W_\pi^\dagger
	=
	A_{\pi^{-1}(1)}\otimes A_{\pi^{-1}(2)}\otimes\cdots\otimes A_{\pi^{-1}(\ell)}.
	$
	Taking the partial trace over the $\alpha$th subsystem yields
	\begin{equation}\label{}
	\mathrm{Tr}_\alpha\!\left[ W_\pi O W_\pi^\dagger \right]
	=
	\mathrm{Tr}\!\left(A_{\pi^{-1}(\alpha)}\right)\cdot 
	\bigotimes_{j\neq \alpha} A_{\pi^{-1}(j)}.
	\end{equation}
	Writing $u=\pi^{-1}(\alpha)$, this becomes
	\begin{equation}\label{lhs}
	\mathrm{Tr}_\alpha\!\left[ W_\pi O W_\pi^\dagger \right]
	=
	\mathrm{Tr}(A_u) \cdot  \bigotimes_{j\neq \alpha} A_{\pi^{-1}(j)}.
	\end{equation}
	
	On the other hand, tracing out the $u$th subsystem of the original operator gives
	\begin{equation}\label{}
	\mathrm{Tr}_u(O)
	=
	\mathrm{Tr}(A_u) \cdot \bigotimes_{k\neq u} A_k.
	\end{equation}
	Since $\pi(u)=\alpha$, the permutation $\pi$ induces a bijection from $\{1,\ldots,\ell\}\setminus\{u\}$ to $\{1,\ldots,\ell\}\setminus\{\alpha\}$. Let $\widetilde{W}_\pi$ denote the corresponding permutation unitary on the remaining $\ell-1$ tensor factors. Its action is
	\begin{equation}\label{}
	\widetilde{W}_\pi
	\left(\bigotimes_{k\neq u} A_k\right)
	\widetilde{W}_\pi^\dagger
	=
	\bigotimes_{j\neq \alpha} A_{\pi^{-1}(j)}.
	\end{equation}
	Therefore,
	\begin{equation}\label{rhs}
	\widetilde{W}_\pi \,\mathrm{Tr}_u(O)\,\widetilde{W}_\pi^\dagger
	=
	\mathrm{Tr}(A_u) \cdot \bigotimes_{j\neq \alpha} A_{\pi^{-1}(j)}.
	\end{equation}
	Comparing Eqs.~\eqref{lhs} and \eqref{rhs}, we obtain
	\begin{equation}\label{}
	\mathrm{Tr}_\alpha\!\left[ W_\pi O W_\pi^\dagger \right]
	=
	\widetilde{W}_\pi \,\mathrm{Tr}_u(O)\,\widetilde{W}_\pi^\dagger.
	\end{equation}
	By linearity, the identity holds for every operator $O$ on $\mathcal{H}^{\otimes \ell}$.
\end{proof}

\begin{lemma}
	\label{lem:definetti-normalization}
	Let $\{\varrho^{(\ell)}\}_{\ell\ge 2}$ be the exchangeable CQ-state set constructed from symmetric feasible solutions $\{\mathrm{\Pi}_{\mathbf{i}}^{(\ell)}\}_{\mathbf{i}\in[m]^\ell}$ of $\mathcal{P}_\ell(d,m)$, and let $\varrho^{(\ell)}=\int d\eta(\varsigma)\,\varsigma^{\otimes \ell}$ be its quantum de Finetti representation, where $\eta$ is supported on single-copy CQ states $\varsigma=\sum_{x=1}^m |x\rangle\langle x|\otimes \varsigma_x$. Then, for $\eta$-almost everywhere, the operators $N_x:=d\,\varsigma_x$ form a valid POVM.
\end{lemma}

\begin{proof}
	By definition of the CQ state,
	\begin{equation}\label{}
	\varrho^{(\ell)}=\frac{1}{d^\ell}\sum_{\mathbf{i}} |\mathbf{i}\rangle\langle\mathbf{i}|_C \otimes  {\mathrm{\Pi}}_{\mathbf{i}}^{(\ell)}.
	\end{equation}
	Taking the partial trace over all classical registers $C_1,\dots,C_\ell$, we obtain the quantum marginal
	$\varrho_Q^{(\ell)}:=\mathrm{Tr}_C(\varrho^{(\ell)})=\frac{1}{d^\ell}\sum_{\mathbf{i}} {\mathrm{\Pi}}_{\mathbf{i}}^{(\ell)}$.
	By the SDP constraint $\sum_{\mathbf{i}\setminus i_j} {\mathrm{\Pi}}_{\mathbf{i}}^{(\ell)}=\mathbb{I}_{d^{j-1}}\otimes M_{i_j}\otimes \mathbb{I}_{d^{\ell-j}}$, summing further over $i_j$ and using the completeness relation $\sum_x M_x=\mathbb{I}_d$, we obtain $\sum_{\mathbf{i}} {\mathrm{\Pi}}_{\mathbf{i}}^{(\ell)}=\mathbb{I}_{d^\ell}$. Therefore, for every $\ell\ge 2$,
	\begin{equation}\label{}
	\varrho_Q^{(\ell)}=\frac{\mathbb{I}_{d^\ell}}{d^\ell}=\left(\frac{\mathbb{I}_d}{d}\right)^{\otimes \ell}.
	\end{equation}
	
	Now consider the quantum de Finetti representation of the exchangeable limit of $\varrho^{(\ell)}$, namely $\varrho^{(\ell)}=\int d\eta(\varsigma)\,\varsigma^{\otimes \ell}$. Applying the partial trace only to the quantum subsystem, and defining the one-copy quantum marginal by $\varsigma_Q:=\mathrm{Tr}_C(\varsigma)=\sum_{x=1}^m \varsigma_x$, we obtain
	\begin{equation}\label{}
	\int d\eta(\varsigma)\,\varsigma_Q^{\otimes \ell}
	=
	\left(\frac{\mathbb{I}_d}{d}\right)^{\otimes \ell},
	\quad \forall\,\ell.
	\end{equation}
	In probability and measure theory~\cite{1995billingsleyProbabilityMeasure}, if all tensor moments of a distribution coincide with the corresponding tensor powers of its mean, then the distribution must be a Dirac measure concentrated at that mean. Hence the support of $\eta$ is contained in states satisfying $\varsigma_Q=\mathbb{I}_d/d$ $(\eta\text{-a.e.})$. Substituting the definition $\varsigma_Q=\sum_{x=1}^m \varsigma_x$, we obtain $\sum_{x=1}^m \varsigma_x=\mathbb{I}_d/d$ $(\eta\text{-a.e.})$. Now define $N_x:=d\,\varsigma_x$. Multiplying the identity $\sum_{x=1}^m \varsigma_x=\mathbb{I}_d/d$ by $d$, we obtain $\sum_{x=1}^m N_x=\mathbb{I}_d$ $(\eta\text{-a.e.})$. Finally, since $\varsigma=\sum_{x=1}^m |x\rangle\langle x|\otimes \varsigma_x$ is a CQ state, each block $\varsigma_x$ is positive semidefinite, and therefore $N_x=d\,\varsigma_x\succeq 0$ for all $x\in[m]$. Thus, for $\eta$-a.e., the set $\{N_x\}_{x=1}^m$ forms a valid $m$-outcome POVM.
\end{proof}

\subsection{Proof of Theorem~\protect\hyperlink{theorem1}{1}}

The proof follows the same de Finetti strategy as Lemma~\ref{lem:hierarchy1}. The only new ingredient is that the additional constraints $\mathrm{Tr}(\mathrm{\Pi}_{\mathbf{i}}^{(\ell)})=1$ fix the classical marginal of the exchangeable CQ state, which in turn forces the hidden projective measurement to be rank one.

\begin{proof}[Proof of Theorem~\hyperlink{theorem1}{1}]
	The ``if'' direction is immediate. Suppose that $\mathbf{M}=\{M_i\}_{i=1}^d\in\mathcal{P}(d,d;1)$. By definition, there exist weights $w_k\ge 0$ with $\sum_k w_k=1$ and rank-one projective measurements $\mathbf{P}^{\,k}=\{P_i^k\}_{i=1}^d$ such that $M_i=\sum_k w_k P_i^k$ for all $i\in[d]$. For each $\ell$, define
	\begin{equation}\label{} 
	\mathrm{\Pi}_{\mathbf{i}}^{(\ell)}:=\sum_k w_k\, P_{i_1}^k\otimes\cdots\otimes P_{i_\ell}^k,
	\quad
	\mathbf{i}=(i_1,\dots,i_\ell)\in[d]^\ell.
	\end{equation}
	The verification of the positivity, marginal, and swap constraints is identical to that in the proof of Lemma~\ref{lem:hierarchy1}. Moreover, since each $P_i^k$ is rank one, $\mathrm{Tr}(P_i^k)=1$, and hence
	\begin{equation}\label{} 
	\mathrm{Tr}\!\left(\mathrm{\Pi}_{\mathbf{i}}^{(\ell)}\right)
	=
	\sum_k w_k \prod_{j=1}^\ell \mathrm{Tr}(P_{i_j}^k)
	=
	\sum_k w_k
	=
	1,
	\end{equation}
	for every $\mathbf{i}\in[d]^\ell$. Therefore $\mathbf{M}\in\mathcal{P}_\ell(d,d;1)$ for every $\ell\ge 2$.
	
	We now prove the converse. Assume that $\mathbf{M}=\{M_i\}_{i=1}^d$ belongs to $\mathcal{P}_\ell(d,d;1)$ for every $\ell\ge 2$. Since $ \mathcal{P}_\ell(d,d;1)\subseteq\mathcal{P}_\ell(d,d), $ Lemma~\ref{lem:hierarchy1} implies that $\mathbf{M}\in\mathcal{P}(d,d)$. In particular, there exists a probability measure $\eta$ supported on projective measurements such that
	\begin{equation}
		\label{equ:rankone-simul}
		M_x=\int_{\mathbb{M}(d,d)} d\eta(N)\,N_x,
		\quad \forall\,x\in[d].
	\end{equation}
	It remains to show that these projective measurements are in fact rank one.
	
	To this end, we exploit the additional trace-one constraints $\mathrm{Tr}(\mathrm\Pi_{\mathbf{i}}^{(\ell)})=1$. Tracing out the quantum registers of the exchangeable CQ state
	$\rho^{(\ell)}$
	constructed in the proof of Lemma~\ref{lem:hierarchy1}, we obtain
	\begin{equation}
		\begin{aligned}
			\rho_C^{(\ell)}
			:&=
			\mathrm{Tr}_Q(\rho^{(\ell)})\\
			&=
			\frac{1}{d^\ell}
			\sum_{\mathbf{i}\in[d]^\ell}
			\mathrm{Tr}(\mathrm\Pi_{\mathbf{i}}^{(\ell)})
			|\mathbf{i}\rangle\langle\mathbf{i}|_C \\
			&=
			\frac{1}{d^\ell}
			\sum_{\mathbf{i}\in[d]^\ell}
			|\mathbf{i}\rangle\langle\mathbf{i}|_C
			=
			\left(\frac{\mathbb I_d}{d}\right)^{\otimes\ell},
		\end{aligned}
	\end{equation}
	where the third equality uses
	$\mathrm{Tr}(\mathrm\Pi_{\mathbf{i}}^{(\ell)})=1$
	for every $\mathbf{i}$.
	
	On the other hand, for a single-copy CQ state $\varsigma=\sum_{x=1}^d |x\rangle\langle x|\otimes \varsigma_x$, its classical marginal is
	$
	\varsigma_C:=\mathrm{Tr}_Q(\varsigma)=\sum_{x=1}^d \mathrm{Tr}(\varsigma_x)\,|x\rangle\langle x|.
	$
	Taking the partial trace over the quantum subsystem in the quantum de Finetti representation gives
	\begin{equation}\label{} 
		\int d\eta(\varsigma)\,\varsigma_C^{\otimes \ell}
		=
		\left(\frac{\mathbb{I}_d}{d}\right)^{\otimes \ell},
		\quad \forall\,\ell.
	\end{equation}
	By the same moment argument used in Lemma~\ref{lem:definetti-normalization}, we obtain
	\begin{equation}\label{} 
		\varsigma_C=\frac{\mathbb{I}_d}{d}
		\quad
		(\eta\text{-a.e.}).
	\end{equation}
	Equivalently,
	\begin{equation}\label{} 
		\mathrm{Tr}(\varsigma_x)=\frac{1}{d},
		\quad \forall\,x\in[d],
		\quad
		(\eta\text{-a.e.}).
	\end{equation}
	Since $N_x=d\,\varsigma_x$, we conclude that
	\begin{equation}\label{} 
		\mathrm{Tr}(N_x)=1,
		\quad \forall\,x\in[d],
		\quad
		(\eta\text{-a.e.}).
	\end{equation}
	Therefore each $N_x$ is an orthogonal projector with unit trace, and therefore is rank one. Consequently, the measure $\eta$ is supported on rank-one projective measurements. Combining this with Eq.~\eqref{equ:rankone-simul}, we conclude that $\mathbf{M}$ is a convex mixture of rank-one projective measurements, i.e., $ \mathbf{M}\in\mathcal{P}(d,d;1)$.

\end{proof}

\setcounter{equation}{0}
\renewcommand\theequation{C\arabic{equation}}

\section{Numerical Characterization of $\mathcal{C}_\ell(d,n)$ and Computational Complexity}
\label{appC}

In this Appendix, we discuss the numerical implementation of the SDP hierarchy for $\mathcal{C}_\ell(d,n)$ and the associated computational tradeoffs. The starting point is the decomposition
$
\tau_{i,\mu}:=\int d\lambda\, p(\lambda)q(\mu|\lambda)\,|e_i^\lambda\rangle\langle e_i^\lambda|
$
or, equivalently, its normalized form $\tau_{i|\mu}$ given in Eq.~\eqref{equ:normalizedoperator}. A key structural feature is that the rank-one projectors $\{|e_i^\lambda\rangle\langle e_i^\lambda|\}_{i=1}^d$ depend on the hidden index $\lambda$, but not on $\mu$. Thus, different deterministic assignments $\mu$ modify the mixing weights, but not the underlying set of $d$-outcome rank-one projective measurements.

This observation leads to two implementations of the SDP hierarchy. The first is the classical simulability hierarchy used in the main text, which imposes the lifted constraints separately for each deterministic assignment $\mu$ and introduces no additional relaxation beyond the level-$\ell$ approximation itself. The second is the aggregated hierarchy, obtained by first aggregating over $\mu$ and then imposing the lifted constraints only on a single POVM. The latter is computationally cheaper but yields a weaker outer approximation. We discuss the two approaches in turn.

\subsection{Classical Simulability Hierarchy}
\label{appC:canonical}

Suppose that, for each deterministic assignment $\mu$, we are given a $d$-outcome POVM
$
\boldsymbol{\tau}_{\mu}
=
\{\tau_{i|\mu}\}_{i=1}^d,
$
satisfying
$
\mathrm{Tr}(\tau_{i|\mu})=1,
$
$
\sum_{i=1}^d\tau_{i|\mu}=\mathbb{I}_d,
$
and
$
\tau_{i|\mu}\succeq 0.
$
The relevant question is whether these POVMs can all be simulated by a common set of rank-one projective measurements, namely whether there exist weights $f(\lambda|\mu)$ such that
\begin{equation}
	\label{equ:projectivesimu}
	\tau_{i|\mu}
	=
	\int d \lambda \, f(\lambda|\mu)\,P_i^\lambda,
	\quad \forall\, i,\mu,
\end{equation}
where, for every $\lambda$, the set $\{P_i^\lambda\}_{i=1}^d$ is a rank-one projective measurement.

At first sight, Eq.~\eqref{equ:projectivesimu} appears stronger than testing each $\mu$ separately, because it asks the same rank-one projective measurements to work for all $\mu$. This causes no loss of generality. If, for each $\mu$, one has an independent decomposition into rank-one projective measurements, then one may take the union of all projective measurements appearing across different $\mu$ and assign zero weight to those not used by a given deterministic assignment. For example, if
$
\boldsymbol{\tau}_{1}
=
w_{1|1}\mathbf{P}^{1}
+
w_{2|1}\mathbf{P}^{2}
$
and
$
\boldsymbol{\tau}_{2}
=
w_{2|2}\mathbf{P}^{2}
+
w_{3|2}\mathbf{P}^{3},
$
then both can be rewritten using the common list
$
\{\mathbf{P}^{1},\mathbf{P}^{2},\mathbf{P}^{3}\}
$
by setting
$
(f_{\lambda|1})_\lambda
=
(w_{1|1},w_{2|1},0)
$
and
$
(f_{\lambda|2})_\lambda
=
(0,w_{2|2},w_{3|2}),
$
with
$
\sum_\lambda f_{\lambda|1}
=
\sum_\lambda f_{\lambda|2}
=
1.
$
Hence, testing projective simulability for each $\mu$ separately already guarantees the existence of a common rank-one projective list after relabeling.

This observation gives rise to the classical simulability hierarchy: for each deterministic assignment $\mu$, one imposes the level-$\ell$ lifted constraints characterizing membership in $\mathcal{P}_\ell(d,d;1)$. Concretely, if there exist lifted operators $\{\mathrm{\Pi}_{\mathbf{i},\mu}^{(\ell)}\}_{\mathbf{i},\mu}$ such that
\begin{equation}
	\label{equ:strict-appC}
	\left\{
	\begin{aligned}
		& \mathrm{\Pi}_{\mathbf{i},\mu}^{(\ell)}\succeq 0,\quad
		\mathrm{Tr}(\mathrm{\Pi}_{\mathbf{i},\mu}^{(\ell)})=1,
		\quad \forall\,\mathbf{i},\mu,\\
		& \sum_{\mathbf{i}\setminus i_j}\mathrm{\Pi}_{\mathbf{i},\mu}^{(\ell)}
		=\mathbb{I}_{d^{j-1}}\otimes \tau_{i_j|\mu}\otimes \mathbb{I}_{d^{\ell-j}},
		\quad \forall\, j,\ i_j,\mu,\\
		& \mathrm{Tr}_\alpha\!\left(
		V_{(\alpha\beta)}\mathrm{\Pi}_{\mathbf{i},\mu}^{(\ell)}
		\right)=0,
		\quad \forall\,\mu,\ 1\le\alpha<\beta\le\ell,\ 
		\mathbf{i}\ \text{with}\ i_\alpha\neq i_\beta,
	\end{aligned}
	\right.
\end{equation}
then
$
\{\tau_{i|\mu}\}_{i=1}^d
\in
\mathcal{P}_\ell(d,d;1)
$
for every deterministic assignment $\mu$. This is precisely the membership test used in the main text, after passing from normalized POVM elements $\tau_{i|\mu}$ to their conic representatives $\tau_{i,\mu}=c_\mu\tau_{i|\mu}$. Thus, Eq.~\eqref{equ:mainsdp} is the conic version of imposing Eq.~\eqref{equ:strict-appC} separately for each deterministic assignment.

The advantage of the classical simulability hierarchy is that it introduces no additional relaxation beyond the level-$\ell$ relaxation of $\mathcal P(d,d;1)$. Its drawback is computational cost. The lifted variables are indexed by both $\mathbf{i}\in[d]^\ell$ and $\mu\in[d]^n$, so the number of positive semidefinite blocks is $d^{n+\ell}$, each of matrix size $d^\ell\times d^\ell$. Consequently, the total number of matrix elements associated with the lifted variables scales as
$$
O(d^n)\times O(d^\ell)\times O(d^{2\ell})
=
O(d^{\,n+3\ell}).
$$
In addition, the variables $\tau_{i,\mu}$ contribute $d^{n+1}$ positive semidefinite $d\times d$ matrices, i.e.,
$
O(d^{\,n+3})
$
scalar variables, which is typically negligible compared with the lifted variables once $\ell\ge 2$. Thus the dominant computational cost arises from introducing a separate collection of lifted variables for every deterministic assignment.

\subsection{Aggregated Hierarchy}
\label{appC:aggregated}

A cheaper alternative follows from the fact that each conic operator
$\tau_{i,\mu}=c_\mu\tau_{i|\mu}$
belongs to the convex cone generated by rank-one projectively simulable POVMs. Since the cone is convex, the sum over $\mu$ remains in the same cone. Define
$$
\bar{\tau}_i:=\sum_{\mu}\tau_{i,\mu}
=
\sum_\mu c_\mu\tau_{i|\mu}.
$$
Using a rank-one projective simulation model,
\[
\tau_{i,\mu}
=
\int d\lambda\,
p(\lambda)\,
q(\mu|\lambda)\,
|e_i^\lambda\rangle\langle e_i^\lambda|,
\]
we obtain
$$
\bar{\tau}_i
=
\int d\lambda\,p(\lambda)
\Bigl(\sum_\mu q(\mu|\lambda)\Bigr)
|e_i^\lambda\rangle\langle e_i^\lambda|
=
\int d\lambda\,p(\lambda)
|e_i^\lambda\rangle\langle e_i^\lambda|,
$$
because $\sum_\mu q(\mu|\lambda)=1$. Hence the aggregated operators
$\{\bar{\tau}_i\}_{i=1}^d$
form a $d$-outcome POVM simulable by rank-one projective measurements, i.e.,
$\{\bar{\tau}_i\}_{i=1}^d\in\mathcal{P}(d,d;1)$.

This motivates the aggregated hierarchy, in which the lifted constraints are imposed only once, on the aggregate POVM
$\bar{\tau}:=\{\bar{\tau}_i\}_{i=1}^d$,
rather than separately for every $\mu$. At level $\ell$, one may introduce variables
$\{\tau_{i,\mu}\}_{i,\mu}$,
$\{c_\mu\}_\mu$,
$\{\bar{\tau}_i\}_i$,
and a single lifted set
$\{\bar{\mathrm{\Pi}}_{\mathbf{i}}^{(\ell)}\}_{\mathbf{i}\in[d]^\ell}$
satisfying
\begin{equation}
	\label{equ:sumrelax-appC}
	\left\{
	\begin{aligned}
		& \rho_x=\sum_{\mu,i}
		D(i|x,\mu)\,\tau_{i,\mu},
		\quad \forall\,x,\\
		& \bar{\tau}_i=\sum_\mu \tau_{i,\mu},
		\quad \forall\,i,\\
		& \tau_{i,\mu}\succeq 0,\quad
		\mathrm{Tr}(\tau_{i,\mu})=c_\mu,
		\quad \forall\,i,\mu,\\
		& c_\mu\ge 0,\quad
		\sum_{i=1}^{d}\tau_{i,\mu}
		=
		c_\mu\mathbb{I}_d,
		\quad \forall\,\mu,\\
		& \sum_\mu c_\mu=1,\\
		& \bar{\mathrm{\Pi}}_{\mathbf{i}}^{(\ell)}\succeq 0,\quad
		\mathrm{Tr}(\bar{\mathrm{\Pi}}_{\mathbf{i}}^{(\ell)})=1,
		\quad \forall\,\mathbf{i},\\
		& \sum_{\mathbf{i}\setminus i_j}
		\bar{\mathrm{\Pi}}_{\mathbf{i}}^{(\ell)}
		=
		\mathbb{I}_{d^{j-1}}
		\otimes
		\bar{\tau}_{i_j}
		\otimes
		\mathbb{I}_{d^{\ell-j}},
		\quad \forall\,j\in[\ell],\ i_j\in[d],\\
		& \mathrm{Tr}_\alpha\!\left(
		V_{(\alpha\beta)}
		\bar{\mathrm{\Pi}}_{\mathbf{i}}^{(\ell)}
		\right)=0,
		\quad
		\forall\,1\le\alpha<\beta\le\ell,\ 
		\mathbf{i}\ \text{with}\ i_\alpha\neq i_\beta .
	\end{aligned}
	\right.
\end{equation}

Any feasible point of the classical simulability hierarchy induces a feasible point of  Eq.~\eqref{equ:sumrelax-appC} by aggregating over the deterministic assignment index $\mu$. Therefore, the aggregated hierarchy defines a larger feasible region. Consequently, infeasibility of Eq.~\eqref{equ:sumrelax-appC} still rules out feasibility of the classical simulability hierarchy, and hence rules out classical simulability; however, feasibility of the aggregated hierarchy does not imply feasibility of the classical simulability hierarchy.

The main benefit of Eq.~\eqref{equ:sumrelax-appC} is computational. The lifted part now consists of only $d^\ell$ positive semidefinite blocks of size $d^\ell\times d^\ell$, so the dominant number of scalar matrix entries drops to
$$
O(d^\ell)\times O(d^{2\ell})
=
O(d^{\,3\ell}),
$$
removing the exponential prefactor $d^n$ from the lifted variables. The variables  $\tau_{i,\mu}$ still contribute $O(d^{\,n+3})$ scalar variables. Consequently, compared with the classical simulability hierarchy, the aggregated hierarchy is substantially cheaper in both memory and runtime, especially when $n$ is moderate or large.

The two hierarchies therefore represent different tradeoffs between computational cost and relaxation strength. The classical simulability hierarchy imposes the lifted constraints separately for each deterministic assignment and therefore introduces no additional relaxation beyond the level-$\ell$ approximation of $\mathcal{P}(d,d;1)$. Its dominant lifted-variable complexity scales as $ O(d^{\,n+3\ell}).$ By contrast, the aggregated hierarchy first sums over the deterministic assignments and then imposes the lifted constraints only on the aggregated POVM. This reduces the dominant lifted-variable complexity to $ O(d^{\,3\ell}), $ at the expense of yielding a weaker outer approximation. Consequently, the classical simulability hierarchy provides a tighter level-$\ell$ outer approximation, whereas the aggregated hierarchy offers a significantly more efficient alternative for larger-scale computations.

\setcounter{equation}{0}
\renewcommand\theequation{D\arabic{equation}}

\section{Derivation of the Dual Form}
\label{appD}
In this Appendix, we derive the dual SDP corresponding to the primal feasibility problem in Eq.~\eqref{equ:mainsdp}. Throughout, the hierarchy level $\ell$ is fixed, and, to lighten notation, we omit the superscript $(\ell)$ on the lifted operators. We then work with the primal form
\begin{equation}
	\label{equ:primal-appD}
	\begin{aligned}
		\max\quad & 0 \\
		\mathrm{w.r.t.}\quad& {\{\tau_{i,\mu}\}_{i,\mu},\, \{c_\mu\}_\mu,\, \{\mathrm{\Pi}_{\mathbf{i},\mu}\}_{\mathbf{i},\mu}}, \\
		\mathrm{s.t.}\quad 
		& \rho_x=\sum_{\mu,i} D(i|x,\mu)\,\tau_{i,\mu},\quad \forall\,x,\\
		& \sum_{\mathbf{i}\setminus i_j}\mathrm{\Pi}_{\mathbf{i},\mu}
		=\mathbb{I}_{d^{j-1}}\otimes \tau_{i_j,\mu}\otimes \mathbb{I}_{d^{\ell-j}},
		\quad \forall\,j,\ i_j,\mu,\\
		& \mathrm{Tr}(\mathrm{\Pi}_{\mathbf{i},\mu})=c_\mu,\quad \forall\,\mathbf{i},\mu,\\
		& \mathrm{Tr}_\alpha\!\left(V_{(\alpha\beta)}\mathrm{\Pi}_{\mathbf{i},\mu}\right)=0,
		\quad \forall\,\mu,\ 1\le\alpha<\beta\le\ell,\ \mathbf{i}\ \text{with}\ i_\alpha\neq i_\beta,\\
		& \sum_{i} \tau_{i,\mu} = c_\mu \mathbb{I}_d, \quad\forall\,\mu,\\
		& \sum_{\mu} c_\mu=1,\\
		& \tau_{i,\mu}\succeq 0,\quad \mathrm{\Pi}_{\mathbf{i},\mu}\succeq 0,\quad c_\mu\ge 0,\quad \forall\,i,\mathbf{i},\mu.
	\end{aligned}
\end{equation}
This formulation is equivalent to the SDP in Eq.~\eqref{equ:mainsdp}; in particular, the scalar constraints $\mathrm{Tr}(\tau_{i,\mu})=c_\mu$ already follow from $\mathrm{Tr}(\mathrm{\Pi}_{\mathbf{i},\mu})=c_\mu$ and the marginal constraints. To each equality constraint in Eq.~\eqref{equ:primal-appD} we assign a dual variable:
\begin{equation}
	\label{equ:dualvars-appD}
	\left\{
	\begin{aligned}
		& W_x\in\mathbb{H}_d
		&& \text{for } \rho_x-\sum_{\mu,i}D(i|x,\mu)\tau_{i,\mu}=0,\\
		& Z_{j,s,\mu}\in\mathbb{H}_{d^\ell}
		&& \text{for } \sum_{\mathbf{i}\setminus i_j=s}\mathrm{\Pi}_{\mathbf{i},\mu}
		-\mathbb{I}_{d^{j-1}}\otimes \tau_{s,\mu}\otimes \mathbb{I}_{d^{\ell-j}}=0,\\
		& y_{\mathbf{i},\mu}\in\mathbb{R}
		&& \text{for } \mathrm{Tr}(\mathrm{\Pi}_{\mathbf{i},\mu})-c_\mu=0,\\
		& Y_{\alpha\beta,\mathbf{i},\mu}\in\mathbb{C}^{d^{\ell-1}\times d^{\ell-1}}
		&& \text{for } \mathrm{Tr}_\alpha\!\left(V_{(\alpha\beta)}\mathrm{\Pi}_{\mathbf{i},\mu}\right)=0,\\
		&  S_\mu \in \mathbb{H}_d
		&& \text{for } \sum_{i} \tau_{i,\mu} - c_\mu \mathbb{I}_d = 0, \\
		& t\in\mathbb{R}
		&& \text{for } \sum_\mu c_\mu-1=0.
	\end{aligned}
	\right.
\end{equation}

The choice of a complex multiplier $Y_{\alpha\beta,\mathbf{i},\mu}$ requires a brief explanation. Although both $\mathrm{\Pi}_{\mathbf{i},\mu}$ and $V_{(\alpha\beta)}$ are Hermitian, the product $V_{(\alpha\beta)}\mathrm{\Pi}_{\mathbf{i},\mu}$ is not Hermitian in general, and therefore neither is its partial trace $\mathrm{Tr}_\alpha(V_{(\alpha\beta)}\mathrm{\Pi}_{\mathbf{i},\mu})$. To enforce the vanishing of this matrix constraint, one must constrain both its Hermitian and anti-Hermitian parts. This is achieved by pairing the constraint with an unconstrained complex multiplier and adding the term together with its Hermitian conjugate (h.c.), thereby ensuring that the Lagrangian remains real. 

With the convention that each multiplier couples to ``left-hand side minus right-hand side,'' the Lagrangian reads
\begin{equation}
	\label{equ:lagrangian-appD}
	\begin{aligned}
		\mathcal{L}
		={}&
		\sum_x \mathrm{Tr}\!\left[
		W_x\left(
		\rho_x-\sum_{\mu,i} D(i|x,\mu)\tau_{i,\mu}
		\right)\right] \\
		&+
		\sum_{j,s,\mu}
		\mathrm{Tr}\!\left[
		Z_{j,s,\mu}
		\left(
		\sum_{\mathbf{i}\setminus i_j=s}\mathrm{\Pi}_{\mathbf{i},\mu}
		-\mathbb{I}_{d^{j-1}}\otimes \tau_{s,\mu}\otimes \mathbb{I}_{d^{\ell-j}}
		\right)
		\right] \\
		&+
		\sum_{\mathbf{i},\mu}
		y_{\mathbf{i},\mu}
		\left(
		\mathrm{Tr}(\mathrm{\Pi}_{\mathbf{i},\mu})-c_\mu
		\right) \\
		&+
		\sum_{\mathbf{i},\mu}
		\sum_{\substack{\alpha<\beta\\ i_\alpha\neq i_\beta}}
		\left(
		\mathrm{Tr}\!\left[
		Y_{\alpha\beta,\mathbf{i},\mu}\,
		\mathrm{Tr}_\alpha\!\left(V_{(\alpha\beta)}\mathrm{\Pi}_{\mathbf{i},\mu}\right)
		\right]
		+\mathrm{h.c.}
		\right) \\
		&+
		 \sum_{\mu} \mathrm{Tr}\left[ S_\mu \left( \sum_{i} \tau_{i,\mu} - c_\mu \mathbb{I}_d \right) \right] \\
		&+
		t\left(\sum_\mu c_\mu-1\right).
	\end{aligned}
\end{equation}

To rewrite the swap term in a form linear in $\mathrm{\Pi}_{\mathbf{i},\mu}$, we use the identity
\begin{equation}
	\label{equ:partial-trace-dual-identity}
	\mathrm{Tr}\!\left[ Y\,\mathrm{Tr}_\alpha(X)\right]
	=
	\mathrm{Tr}\!\left[(\mathbb{I}_\alpha\otimes Y)\,X\right],
\end{equation}
which is valid for any operator $X$ on $\mathcal{H}_\alpha\otimes \mathcal{H}_{\bar\alpha}$ and any operator $Y$ on $\mathcal{H}_{\bar\alpha}$. Applying Eq.~\eqref{equ:partial-trace-dual-identity} to the swap term gives
\begin{equation}
	\label{equ:D-term-appD}
	\mathrm{Tr}\!\left[
	Y_{\alpha\beta,\mathbf{i},\mu}\,
	\mathrm{Tr}_\alpha\!\left(V_{(\alpha\beta)}\mathrm{\Pi}_{\mathbf{i},\mu}\right)
	\right]
	+\mathrm{h.c.}
	=
	\mathrm{Tr}\!\left[
	\mathcal{D}_{\alpha\beta,\mathbf{i},\mu}\,\mathrm{\Pi}_{\mathbf{i},\mu}
	\right],
\end{equation}
where
$
	\mathcal{D}_{\alpha\beta,\mathbf{i},\mu}
	:=
	(\mathbb{I}_d^\alpha\otimes Y_{\alpha\beta,\mathbf{i},\mu})\,V_{(\alpha\beta)}
	+
	V_{(\alpha\beta)}\,(\mathbb{I}_d^\alpha\otimes Y_{\alpha\beta,\mathbf{i},\mu}^\dagger)
$
is Hermitian. Summing over all relevant pairs $(\alpha,\beta)$, we define
\begin{equation}
	\label{equ:D-total-appD}
	\mathcal{D}_{\mathbf{i},\mu}
	:=
	\sum_{\substack{\alpha<\beta\\ i_\alpha\neq i_\beta}}
	\mathcal{D}_{\alpha\beta,\mathbf{i},\mu}.
\end{equation}

For the marginal terms, we use
\begin{equation}
	\label{equ:Zj-marginal-identity}
	\mathrm{Tr}\!\left[
	Z_{j,s,\mu}
	\left(
	\mathbb{I}_{d^{j-1}}\otimes \tau_{s,\mu}\otimes \mathbb{I}_{d^{\ell-j}}
	\right)
	\right]
	=
	\mathrm{Tr}\!\left[
	\mathrm{Tr}_{\setminus j}(Z_{j,s,\mu})\,\tau_{s,\mu}
	\right],
\end{equation}
where $\mathrm{Tr}_{\setminus j}$ denotes the partial trace over all tensor factors except the $j$th one.

Using Eqs.~\eqref{equ:D-term-appD}--\eqref{equ:Zj-marginal-identity}, the Lagrangian can be regrouped according to the primal variables $\tau_{i,\mu}$, $\mathrm{\Pi}_{\mathbf{i},\mu}$, and $c_\mu$:
\begin{equation}
	\label{equ:lagrangian-grouped-appD}
	\begin{aligned}
		\mathcal{L}
		={}&
		\sum_x \mathrm{Tr}(W_x\rho_x)-t \\
		&+
		\sum_{\mu,i}
		\mathrm{Tr}\!\left[
		\tau_{i,\mu}
		\left(
		-\sum_x D(i|x,\mu)\,W_x
		-\sum_{j=1}^{\ell}\mathrm{Tr}_{\setminus j}(Z_{j,i,\mu})
		+ S_\mu
		\right)
		\right] \\
		&+
		\sum_{\mathbf{i},\mu}
		\mathrm{Tr}\!\left[
		\mathrm{\Pi}_{\mathbf{i},\mu}
		\left(
		\sum_{j=1}^{\ell} Z_{j,i_j,\mu}
		+y_{\mathbf{i},\mu}\,\mathbb{I}_{d^\ell}
		+\mathcal{D}_{\mathbf{i},\mu}
		\right)
		\right] \\
		&+
		\sum_\mu
		c_\mu
		\left(-\mathrm{Tr}(S_\mu)
		-\sum_{\mathbf{i}} y_{\mathbf{i},\mu}+t
		\right).
	\end{aligned}
\end{equation}

The dual function is obtained by maximizing Eq.~\eqref{equ:lagrangian-grouped-appD} over the primal cone $\tau_{i,\mu}\succeq 0$, $\mathrm{\Pi}_{\mathbf{i},\mu}\succeq 0$, and $c_\mu\ge 0$. This supremum is finite if and only if the coefficients of the primal variables obey the corresponding sign conditions:
\begin{equation}
	\label{equ:tau-coeff-appD}
	-\sum_x D(i|x,\mu)\,W_x
	-\sum_{j=1}^{\ell}\mathrm{Tr}_{\setminus j}(Z_{j,i,\mu})+ S_\mu
	\preceq 0,
	\quad \forall\,i,\mu,
\end{equation}
\begin{equation}
	\label{equ:Pi-coeff-appD}
	\sum_{j=1}^{\ell} Z_{j,i_j,\mu}
	+y_{\mathbf{i},\mu}\,\mathbb{I}_{d^\ell}
	+\mathcal{D}_{\mathbf{i},\mu}
	\preceq 0,
	\quad \forall\,\mathbf{i},\mu,
\end{equation}
and
\begin{equation}
	\label{equ:c-coeff-appD}
	-\mathrm{Tr}(S_\mu)-\sum_{\mathbf{i}} y_{\mathbf{i},\mu}+t \le 0,
	\quad \forall\,\mu.
\end{equation}
Since $D(i|x,\mu)=\delta_{i,\mu(x)}$, Eq.~\eqref{equ:tau-coeff-appD} may equivalently be written as
\begin{equation}
	\label{equ:tau-coeff-appD-alt}
	\sum_{x:\mu(x)=i} W_x
	+\sum_{j=1}^{\ell}\mathrm{Tr}_{\setminus j}(Z_{j,i,\mu}) - S_\mu
	\succeq 0,
	\quad \forall\,i,\mu.
\end{equation}

Under these conditions, the supremum of the Lagrangian over the primal variables is simply the constant term $\sum_x \mathrm{Tr}(W_x\rho_x)-t$. We therefore obtain the dual form
\begin{equation}
	\label{equ:dual-appD}
	\begin{aligned}
		\min \quad &
		\sum_x \mathrm{Tr}(W_x\rho_x)-t \\
		\mathrm{s.t.}\quad &
		\sum_{x:\mu(x)=i} W_x
		+\sum_{j=1}^{\ell}\mathrm{Tr}_{\setminus j}(Z_{j,i,\mu}) - S_\mu
		\succeq 0,
		\quad \forall\,i,\mu,\\
		&
		\sum_{j=1}^{\ell} Z_{j,i_j,\mu}
		+y_{\mathbf{i},\mu}\,\mathbb{I}_{d^\ell}
		+\mathcal{D}_{\mathbf{i},\mu}
		\preceq 0,
		\quad \forall\,\mathbf{i},\mu,\\
		&
		\mathrm{Tr}(S_\mu) + \sum_{\mathbf{i}} y_{\mathbf{i},\mu}-t \ge 0,
		\quad \forall\,\mu,
	\end{aligned}
\end{equation}
with dual variables
$
W_x\in\mathbb{H}_d,
$
$
Z_{j,i,\mu}\in\mathbb{H}_{d^\ell},
$
$
S_\mu \in \mathbb{H}_d,
$
$
y_{\mathbf{i},\mu}\in\mathbb{R},
$
$
t\in\mathbb{R},
$
and
$
Y_{\alpha\beta,\mathbf{i},\mu}\in\mathbb{C}^{d^{\ell-1}\times d^{\ell-1}}.
$

\setcounter{equation}{0}
\renewcommand\theequation{E\arabic{equation}}
\section{Numerical Data for Classical Simulations and Dual Witnesses}
\label{appE}

The optimal solutions obtained from Eq.~\eqref{equ:mainsdp} provide numerical values for
$\{c_\mu,\tau_{i,\mu}\}$.
As discussed in the main text, such solutions can be used to construct explicit
classical simulations of the form given in Eq.~\eqref{equ:classical2}, thereby certifying
classical simulability at the corresponding visibility values.
The dual SDP~\eqref{equ:maindualsdp} also yields witness operators
$\{W_x^{(\ell)}\}_x$ together with a threshold
$t^{(\ell)}$
as defined in Eq.~\eqref{equ:dualwitness}. 

In this appendix, we present representative numerical data for the BB84 state set
and for the qutrit state set $\boldsymbol{\rho}_1$ constructed from two MUBs.
Additional numerical results can be found in Ref.~\cite{Li2026ClassicalSimulability}.

\subsection{Classical Simulations}

\paragraph*{BB84 state set.}

For the BB84 state set, we have $d=2$ and $n=4$. The deterministic response functions are indexed by $\mu\in[d]^n=[2]^4$ and are defined by
\[
D(i|x,\mu)
=
\delta_{i,\mu(x)},
\]
where
$\mu(x)\in\{1,2\}$
specifies the deterministic outcome assigned to preparation label
$x$. 

Using optimal level-$2$ solutions of Eq.~\eqref{equ:mainsdp} as guidance, we obtain
\[
c_\mu
=
\left(0,0,0,0,0,\frac14,\frac14,0,
0,\frac14,\frac14,0,0,0,0,0\right),
\]
together with the following nonzero operators:
\[
\begin{aligned}
\tau_{1|6} = \begin{pmatrix} \frac{2-\sqrt{2}}{4} & -\frac{\sqrt{2}}{4} \\ -\frac{\sqrt{2}}{4} & \frac{2+\sqrt{2}}{4} \end{pmatrix}, \quad \tau_{1|7} = \begin{pmatrix} \frac{2+\sqrt{2}}{4} & -\frac{\sqrt{2}}{4} \\ -\frac{\sqrt{2}}{4} & \frac{2-\sqrt{2}}{4} \end{pmatrix}, \quad \tau_{1|10} = \begin{pmatrix} \frac{2-\sqrt{2}}{4} & \frac{\sqrt{2}}{4} \\ \frac{\sqrt{2}}{4} & \frac{2+\sqrt{2}}{4} \end{pmatrix}, \quad \tau_{1|11} = \begin{pmatrix} \frac{2+\sqrt{2}}{4} & \frac{\sqrt{2}}{4} \\ \frac{\sqrt{2}}{4} & \frac{2-\sqrt{2}}{4} \end{pmatrix}, \\ \tau_{2|6} = \begin{pmatrix} \frac{2+\sqrt{2}}{4} & \frac{\sqrt{2}}{4} \\ \frac{\sqrt{2}}{4} & \frac{2-\sqrt{2}}{4} \end{pmatrix}, \quad \tau_{2|7} = \begin{pmatrix} \frac{2-\sqrt{2}}{4} & \frac{\sqrt{2}}{4} \\ \frac{\sqrt{2}}{4} & \frac{2+\sqrt{2}}{4} \end{pmatrix}, \quad \tau_{2|10} = \begin{pmatrix} \frac{2+\sqrt{2}}{4} & -\frac{\sqrt{2}}{4} \\ -\frac{\sqrt{2}}{4} & \frac{2-\sqrt{2}}{4} \end{pmatrix}, \quad \tau_{2|11} = \begin{pmatrix} \frac{2-\sqrt{2}}{4} & -\frac{\sqrt{2}}{4} \\ -\frac{\sqrt{2}}{4} & \frac{2+\sqrt{2}}{4} \end{pmatrix}.
\end{aligned}
\]

For every nonzero deterministic assignment index $\mu\in\{6,7,10,11\}$, the normalized operators $\tau_{i|\mu}=\tau_{i,\mu}/c_\mu$ form a two-outcome POVM. Since $\tau_{1|\mu}$ and $\tau_{2|\mu}$ commute, they can be simultaneously diagonalized and hence are simulable by a rank-one projective measurement. Hence each POVM $\{\tau_{i|\mu}\}_{i=1}^{2}$ belongs to $\mathcal P(2,2;1)$. Substituting the corresponding unnormalized operators $\tau_{i,\mu}=c_\mu\tau_{i|\mu}$ into Eq.~\eqref{equ:classical2} reconstructs the noisy BB84 state set at visibility $v=1/\sqrt{2}$, thereby yielding an explicit classical simulation.
The same procedure applies to the other qubit examples considered in the main text.

\paragraph*{Qutrit state set $\boldsymbol{\rho}_1$.}

For the state set
$\boldsymbol{\rho}_1$,
the optimal solution contains
$36$
nonzero coefficients
$c_\mu$,
all equal to
$1/36$.
The corresponding indices are listed in Ref.~\cite{Li2026ClassicalSimulability}.
As an illustrative example, for
$\mu=141$
one obtains
\[
\begin{aligned}&
\tau_{1|141} = \begin{pmatrix}
	\frac{1}{9} & -\frac{1}{18}+\frac{\sqrt{3}}{18}i & -\frac{2}{9}-\frac{\sqrt{3}}{9}i \\[4pt]
	-\frac{1}{18}-\frac{\sqrt{3}}{18}i & \frac{1}{9} & -\frac{1}{18}+\frac{\sqrt{3}}{6}i \\[4pt]
	-\frac{2}{9}+\frac{\sqrt{3}}{9}i & -\frac{1}{18}-\frac{\sqrt{3}}{6}i & \frac{7}{9}
\end{pmatrix},\\&
\tau_{2|141} = \begin{pmatrix}
	\frac{1}{9} & -\frac{2}{9}-\frac{\sqrt{3}}{9}i & -\frac{1}{18}+\frac{\sqrt{3}}{18}i \\[4pt]
	-\frac{2}{9}+\frac{\sqrt{3}}{9}i & \frac{7}{9} & -\frac{1}{18}-\frac{\sqrt{3}}{6}i \\[4pt]
	-\frac{1}{18}-\frac{\sqrt{3}}{18}i & -\frac{1}{18}+\frac{\sqrt{3}}{6}i & \frac{1}{9}
\end{pmatrix},\\&
\tau_{3|141} = \begin{pmatrix}
	\frac{7}{9} & \frac{5}{18}+\frac{\sqrt{3}}{18}i & \frac{5}{18}+\frac{\sqrt{3}}{18}i \\[4pt]
	\frac{5}{18}-\frac{\sqrt{3}}{18}i & \frac{1}{9} & \frac{1}{9} \\[4pt]
	\frac{5}{18}-\frac{\sqrt{3}}{18}i & \frac{1}{9} & \frac{1}{9}
\end{pmatrix}.
\end{aligned}
\]

For this index, the three operators
$\tau_{1|141}$,
$\tau_{2|141}$,
and
$\tau_{3|141}$
are Hermitian and satisfy
\[
[\tau_{1|141},\tau_{2|141}]
=
[\tau_{1|141},\tau_{3|141}]
=
[\tau_{2|141},\tau_{3|141}]
=
0.
\]
Hence they can be simultaneously diagonalized, and therefore
$\{\tau_{i|141}\}_{i=1}^{3}\in\mathcal P(3,3;1)$.
The same conclusion holds for every $\mu$ with $c_\mu>0$.
Together with the coefficients $c_\mu$, the resulting operators reconstruct the noisy state set at visibility $v=2/3$, thereby providing an explicit classical simulation.

\subsection{Dual Witnesses}

For each state set, solving the dual SDP~\eqref{equ:maindualsdp} yields witness operators
$\{W_x^{(\ell)}\}_x$
and a threshold
$t^{(\ell)}$
defining the affine functional
$\mathcal W^{(\ell)}$
in Eq.~\eqref{equ:dualwitness}.
The corresponding numerical witness operators are listed below.

\paragraph*{BB84 state set.}

The optimal level-$2$ witness is given by
\[
	\begin{aligned}
		W_1^{(2)}
		&=
		\begin{pmatrix}
			-0.1244 & 0 \\
			0 & 0.3760
		\end{pmatrix},
		\quad
		W_2^{(2)}
		=
		\begin{pmatrix}
			0.3750 & 0 \\
			0 & -0.1245
		\end{pmatrix},
		\\
		W_3^{(2)}
		&=
		\begin{pmatrix}
			0.1254 & -0.2500 \\
			-0.2500 & 0.1257
		\end{pmatrix},
		\quad
		W_4^{(2)}
		=
		\begin{pmatrix}
			0.1254 & 0.2500 \\
			0.2500 & 0.1257
		\end{pmatrix},
	\end{aligned}
\]
with
\[
t^{(2)}
=
-0.2050.
\]

\paragraph*{Qutrit state set $\boldsymbol{\rho}_1$.}

The optimal level-$2$ witness is

\[
	\begin{aligned}
		W_1^{(2)}
		&=
		\begin{pmatrix}
			-0.1197 & 0 & 0\\
			0 & 0.1303 & 0\\
			0 & 0 & 0.1303
		\end{pmatrix},
		\quad
		W_2^{(2)}
		=
		\begin{pmatrix}
			0.1302 & 0 & 0\\
			0 & -0.1198 & 0\\
			0 & 0 & 0.1302
		\end{pmatrix},
		\\[2mm]
		W_3^{(2)}
		&=
		\begin{pmatrix}
			0.1302 & 0 & 0\\
			0 & 0.1302 & 0\\
			0 & 0 & -0.1198
		\end{pmatrix},
		\quad
		W_4^{(2)}
		=
		\begin{pmatrix}
			0.0469 & -0.0833 & -0.0833\\
			-0.0833 & 0.0469 & -0.0833\\
			-0.0833 & -0.0833 & 0.0469
		\end{pmatrix},
		\\[2mm]
		W_5^{(2)}
		&=
		\begin{pmatrix}
			0.0470 & 0.0417+0.0722i & 0.0417-0.0722i\\
			0.0417-0.0722i & 0.0470 & 0.0417+0.0722i\\
			0.0417+0.0722i & 0.0417-0.0722i & 0.0470
		\end{pmatrix},
		\\[2mm]
		W_6^{(2)}
		&=
		\begin{pmatrix}
			0.0469 & 0.0417-0.0722i & 0.0417+0.0722i\\
			0.0417+0.0722i & 0.0469 & 0.0417-0.0722i\\
			0.0417-0.0722i & 0.0417+0.0722i & 0.0469
		\end{pmatrix},
	\end{aligned}
\]
with
\[
t^{(2)}=-0.3851.
\]

We present here only the dual witnesses for the two representative state sets discussed above; additional numerical results for the other state sets can be found in Ref.~\cite{Li2026ClassicalSimulability}.

\end{appendix}
\end{document}